\documentclass[11pt]{article}
\usepackage{titling}
\usepackage{booktabs}
\usepackage{array}
\usepackage{tikz}
\usetikzlibrary{arrows.meta}
\usepackage{natbib}
\usepackage{calc}
\usepackage{subcaption}
\usepackage[a4paper, margin=1in]{geometry}
\usepackage{rotating}
\usepackage{booktabs}
\usepackage{mathrsfs}
\usepackage{fancyhdr}
\usepackage{amsfonts}
\usepackage{amstext}
\usepackage{amsmath}
\usepackage{amssymb}
\usepackage{algorithm}
\usepackage{algorithmic}
\usepackage{adjustbox}
\usepackage{hyperref}
\usepackage{bookmark}
\usepackage{setspace}
\usepackage{longtable}
\usepackage{float}

\usepackage{fullpage}
\usepackage{times}
\usepackage[utf8]{inputenc}
\usepackage{units}

\usepackage{nicefrac}
\usepackage{color}
\usepackage{tabularx}
\usepackage{pdflscape}

\newtheorem{theorem}{Theorem}
\newtheorem{assumption}{Assumption}

\newenvironment{proof}[1][Proof]{\noindent \textbf{#1.} }{\  \rule{0.5em}{0.5em}}

\pretitle{\begin{center}\large\bfseries}
\posttitle{\end{center}}
\preauthor{\begin{center}
\large \lineskip 0.5em%
\begin{tabular}[t]{c}}
\postauthor{\end{tabular}\par\end{center}}
\predate{\begin{center}\large}
\postdate{\par\end{center}\vskip 0.5em}

\newcommand{\HRule}{\rule{\linewidth}{0.5mm}}
\newcommand{\BRule}{\rule{\linewidth}{0.25mm}}
\title{\HRule \\[0.1mm]
A Powerful Chi-Square Specification Test with Support Vectors \thanks{Li acknowledges support from the Research Development Fund (RDF-23-02-022) of the Xi'an Jiaotong-Liverpool University. Song acknowledges support from the National Natural Science Foundation of China (Grant Numbers 72373007 and 72333001).}
\\[0.1mm]
\BRule}

\author{Yuhao Li\thanks{yuhao.li@xjtlu.edu.cn} \\ \small Xi'an Jiaotong-Liverpool University
\and
Xiaojun Song\thanks{sxj@gsm.pku.edu.cn} \\ \small Peking University}

\date{} 

\begin{document}

\maketitle

\begin{abstract}
    Specification tests, such as Integrated Conditional Moment (ICM) and Kernel Conditional Moment (KCM) tests, are crucial for model validation but often lack power in finite samples. This paper proposes a novel framework to enhance specification test performance using Support Vector Machines (SVMs) for direction learning. We introduce two alternative SVM-based approaches: one maximizes the discrepancy between nonparametric and parametric classes, while the other maximizes the separation between residuals and the origin. Both approaches lead to a $t$-type test statistic that converges to a standard chi-square distribution under the null hypothesis. Our method is computationally efficient and capable of detecting any arbitrary alternative. Simulation studies demonstrate its superior performance compared to existing methods, particularly in large-dimensional settings.
\end{abstract}

\vspace{1em} 
\noindent \textbf{Keywords:} Classification; Conditional Moment Restrictions; Specification Test; Support Vector Machine.

\newpage \pagenumbering{arabic} \setcounter{page}{1} \setcounter{footnote}{0}

\section{Introduction}
\label{sec:intro}

Consider the following parametric regression model:
\begin{equation*}
	Y = \mathcal{M}_{\theta_0}(X) + \varepsilon_{\theta_0},
	\label{eq:cmr}
\end{equation*}
where $\mathcal{M}_{\theta_0}(X)$ is a parametric specification indexed by an unknown parameter vector $\theta_0 \in \Theta$, with $\Theta \subset \mathbb{R}^q$. Here, $X \in \mathcal{X} \subset \mathbb{R}^q$, $Y \in \mathcal{Y} \subset \mathbb{R}$, and $\varepsilon_{\theta_0}$ represents the parametric error. This framework encompasses many important models, including linear and nonlinear conditional mean regression, quantile regression, treatment effect models, and instrumental variables regressions, among others. Accurate specification of these parametric models is crucial for subsequent statistical inferences.

We aim to test the null hypothesis:
\[
	H_0: \mathbb{P}\left(\mathbb{E}[\varepsilon_{\theta_0} \mid X] = 0\right)=1 \text{ for some } \theta_0\in \Theta
\] 
against the alternative hypothesis:
\[
	H_1: \mathbb{P}\left(\mathbb{E}[\varepsilon_{\theta} \mid X] \neq 0\right) >0 \text{ for all } \theta\in \Theta.
\]
The Integrated Conditional Moment (ICM) framework, introduced by \cite{bierens1982consistent}, provides a classical approach to test the correct specification of the parametric model.

The corresponding test statistic is defined as:
\[
	n\hat{T}_{ICM} = \frac{1}{n}\sum_{i,j = 1}^{n} \varepsilon_{\hat\theta,i} \exp\left(-\frac{\lVert x_i - x_j \rVert_2^2}{2}\right)\varepsilon_{\hat\theta,j},
\]
where $\varepsilon_{\hat\theta,i} = y_i - \mathcal{M}_{\hat{\theta}}(x_i)$, and $\hat{\theta}$ is a consistent estimator for $\theta_0$. Here, $\lVert \cdot \rVert_2$ denotes the Euclidean norm. \cite{muandet2020kernel} extended the ICM test to the Kernel Conditional Moment (KCM) test. By replacing the kernel $\exp(-\lVert x_i - x_j \rVert_2^2 /2)$ in $n\hat{T}_{ICM}$ with any integrally strictly positive definite (ISPD) reproducing kernel $k(x_i,x_j)$, the KCM test statistic is given by:
\[
	n\hat{T}_{KCM} = \frac{1}{n}\sum_{i,j = 1}^{n} \varepsilon_{\hat\theta,i} k(x_i,x_j)\varepsilon_{\hat\theta,j}.
\]
\cite{escanciano2024gaussian} proposed similar test statistics using a Gaussian Process approach, with a focus on addressing the estimation effects arising from $\hat{\theta}$.

Despite the different frameworks used to derive the test statistics mentioned earlier, their shared V-statistic structure suggests a close relationship among them. To unify these statistics under a single framework, we adopt the language of Reproducing Kernel Hilbert Spaces (RKHS). This unified perspective lays the foundation for developing our novel, powerful test statistic.

Given a dataset $\{\varepsilon_{\theta_0,i}, x_i\}_{i=1}^n$, we map the data into an RKHS $\mathcal{H}_k$ with reproducing kernel $k(x, x')$, resulting in points $\varepsilon_{\theta_0,i}k(x_i, \cdot) \in \mathcal{H}_k$. Under mild assumptions on the kernel, the null hypothesis holds if and only if the mean element
\[
\mu_{\theta_0,k} = \mathbb{E}\left[\varepsilon_{\theta_0,1}k(X_1, \cdot)\right] = \boldsymbol{0} \in \mathcal{H}_k \quad \text{(see \citealp{muandet2020kernel})}.
\]
However, directly testing $\mu_{\theta_0,k} = \boldsymbol{0}$ is impractical due to its infinite-dimensional nature. Instead, we project $\mu_{\theta_0,k}$ onto a direction $w \in \mathcal{H}_k$ and analyze the normalized projection:
\[
S_{0,k} = \frac{\left\langle \mu_{\theta_0,k}, w \right\rangle_{\mathcal{H}_k}}{\|w\|_{\mathcal{H}_k}},
\]
where normalization by $\|w\|_{\mathcal{H}_k}$ ensures that the projection values are comparable across different directions. The hypotheses can then be expressed as:
\[
H_0: S_{0,k} = 0 \quad \text{versus} \quad H_1: S_{0,k} \neq 0.
\]

Under the alternative hypothesis $H_1$, the optimal direction $w$ that maximizes $S_{0,k}$ aligns with $\mu_{\theta_0,k}$, i.e., $w^* = \mu_{\theta_0,k}$, leading to $S_{0,k} = \|\mu_{\theta_0,k}\|_{\mathcal{H}_k}$. Consequently, all KCM-type test statistics estimate the squared norm:
\[
S_{0,k}^2 = \left\langle \mu_{\theta_0,k}, \mu_{\theta_0,k} \right\rangle_{\mathcal{H}_k} = \mathbb{E}\left[\varepsilon_{\theta_0} k(X, X') \varepsilon_{\theta_0}'\right],
\]
where random variable $Z'$ is an independent copy of $Z$. Note that the above equality is the direct consequence of the reproducing property:
\[
    \left \langle k(x,\cdot), k(x',\cdot) \right \rangle_{\mathcal{H}_k} = k(x,x').
\]

However, maximizing \(S_{0,k}\) does not necessarily maximize the power of the test, as the power also depends on the variance of the asymptotic null distribution. In the case of KCM-type test statistics, it is
\[
\frac{1}{n}\sum_{i,j = 1}^{n} \varepsilon_{\theta_0,i} k(x_i,x_j)\varepsilon_{\theta_0,j} \overset{d}{\to} \sum_{j=1}^{\infty} \tau_j W_j^2,
\]
where \(W_j \sim \mathcal{N}(0,1)\) and \(\{\tau_j\}\) are the eigenvalues of \(\varepsilon_{\theta_0} k(x,x')\varepsilon_{\theta_0}'\), i.e., they are the solutions of
\[
\tau_j f_j(\varepsilon_{\theta_0},x) = \int \varepsilon_{\theta_0} k(x,x')\varepsilon_{\theta_0}' f_j(\varepsilon_{\theta_0}',x') \, d P(\varepsilon_{\theta_0}',x').
\]
See \cite{muandet2020kernel} for the details. The variance of the asymptotic null distribution is \(V^2 = 2 \sum_{j=1}^{\infty}\tau_j\), and a KCM-type test statistic is powerful only if the signal-to-noise ratio (SNR) \(S^2_{0,k} / V\) is large.

For a given RKHS (i.e., the kernel \(k(\cdot,\cdot)\) is fixed), selecting the optimal direction \(w\) through SNR can be formidable due to the infinite-dimensional nature of the problem. In this paper, we address this power-boosting issue from a novel perspective: we frame the testing problem as a classification problem, and our goal is to learn a direction \(w\) that exhibits desirable separation properties.

Specifically, we interpret the testing problem as two types of classification problems. In the first interpretation, we treat the nonparametric class \(\{\varepsilon_{\theta_0,i} k(x_i,\cdot)\}_{i=1}^n\) as one class and the parametric class \(\{\mathcal{M}_{\theta_0}(x_i) k(x_i,\cdot)\}_{i=1}^n\) as another class. The objective is to find a projection direction \(w\) that maximizes the discrepancy between these two classes. Geometrically, this is equivalent to achieving a large mean difference between the projected two classes while maintaining a relatively small overlap between them. Situations corresponding to a large projected mean difference with significant overlap (resulting in poor power performance) and a moderate projected mean difference with minimal overlap (resulting in good power performance) are illustrated in Figure \ref{fig:separation_illustration}.

\begin{figure}[htbp]
    \centering
    
    \includegraphics[width=\textwidth]{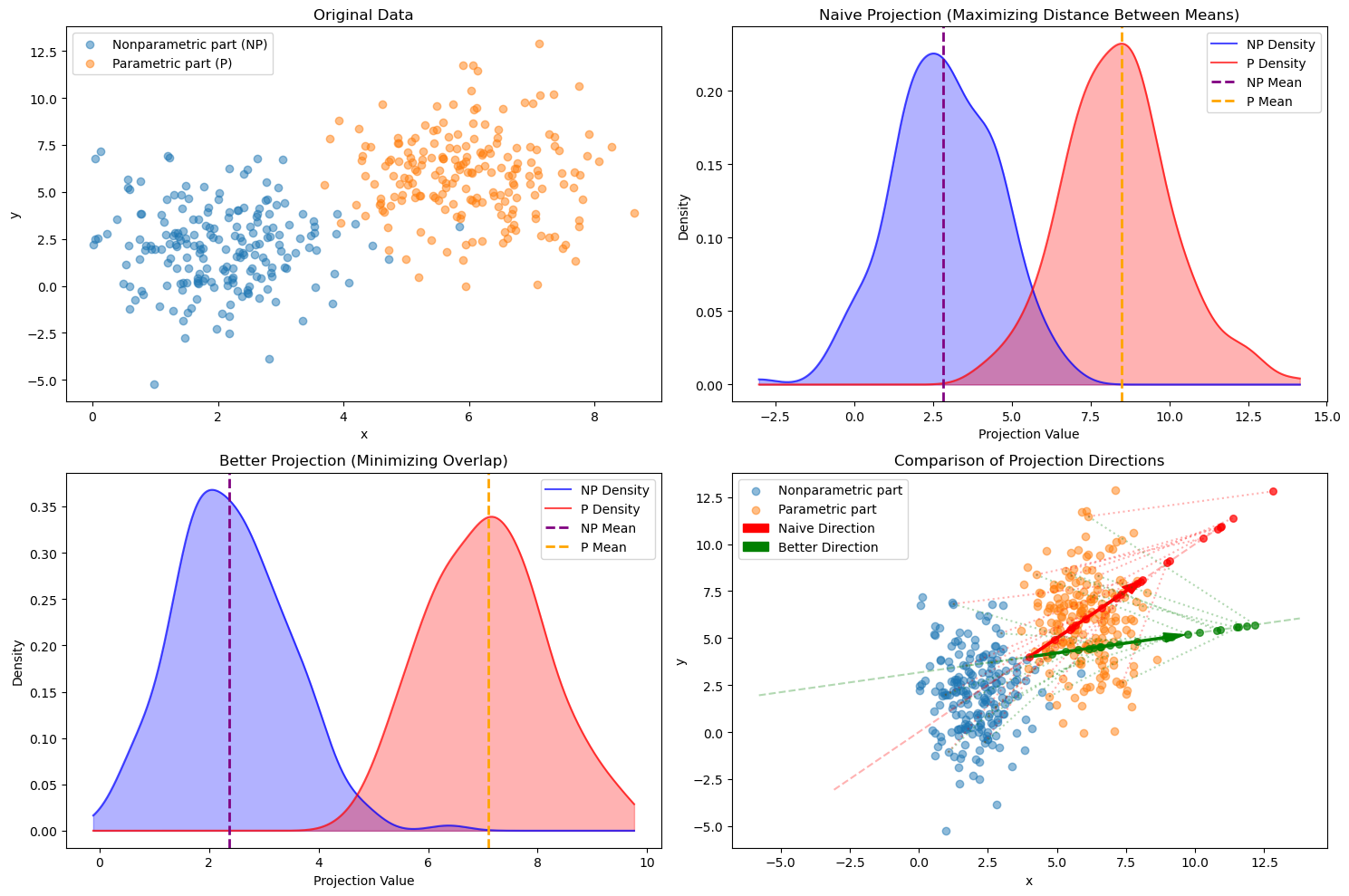}

    \caption{Discrepancy and Separation between the Nonparametric (NP) and Parametric (P) Classes}
    \label{fig:separation_illustration}
\end{figure}

In the second interpretation, we may consider the residuals \(\{\varepsilon_{\theta_0,i} k(x_i,\cdot)\}_{i=1}^n\) as a single class rather than analyzing the separation between two classes. For simplicity, assume the residuals have a positive mean element:
\(\mathbb{E}\left(\langle \varepsilon_{\theta_0,1} k(X_1,\cdot), w \rangle_{\mathcal{H}_k}\right) > 0.\)
An effective projection direction \(w\) should ensure that most projected residuals are greater than zero. That is, even when the mean projection is near zero, only a few residuals should project to negative values. In contrast, a poor projection direction \(w\) (resulting in lower test power) may show many residuals that project below zero, even if their mean projection is far from zero. This reasoning is visualized in Figure \ref{fig:oc_separation_illustration}.
\begin{figure}[htbp]
    \centering
    
    \includegraphics[width=\textwidth]{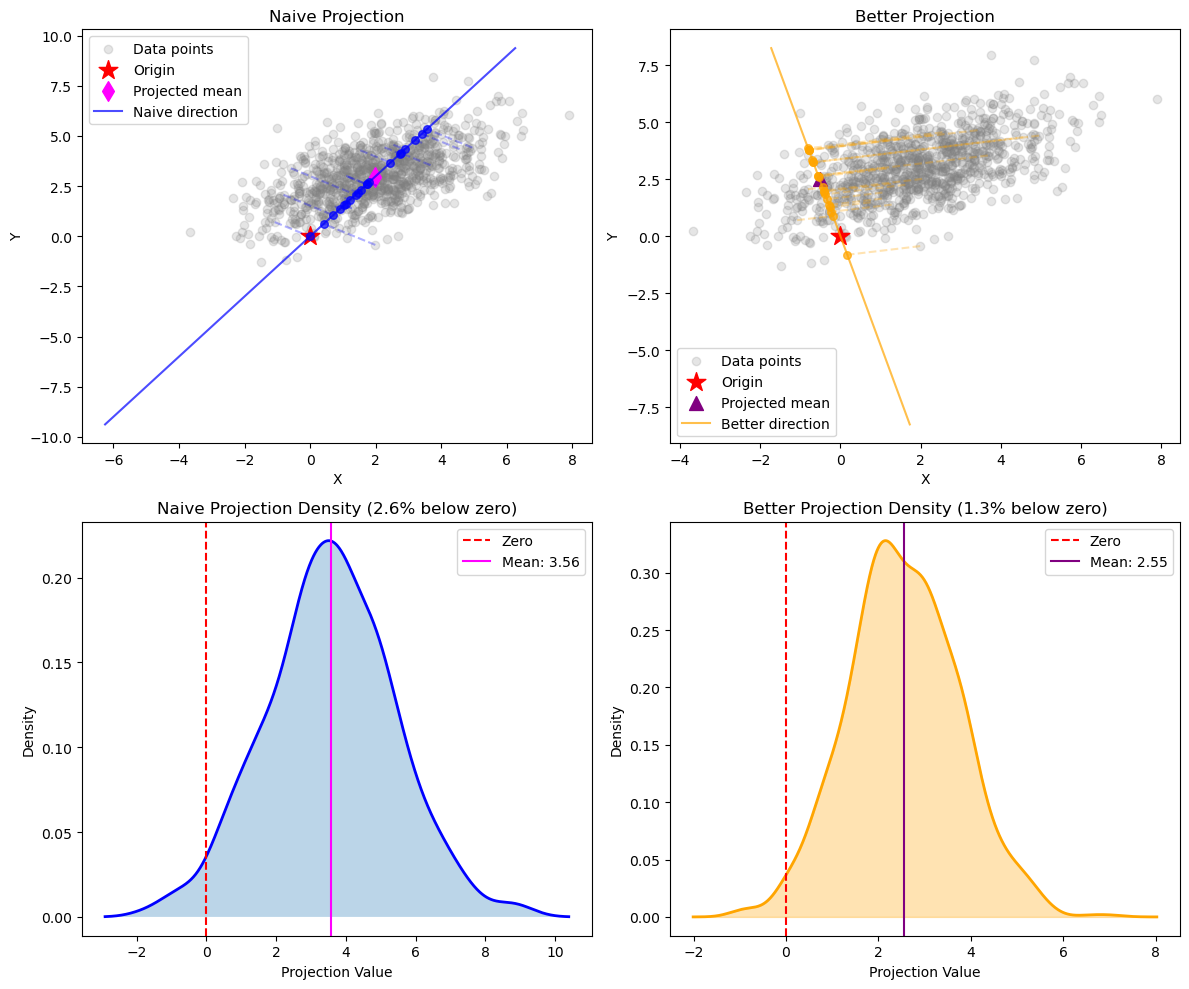}

    \caption{Discrepancy and Separation between the origin point and Residual Class}
    \label{fig:oc_separation_illustration}
\end{figure}

The main contribution of this paper is the application of two Support Vector Machine (SVM) algorithms to a training dataset (that is independent from the dataset used for testing) to learn a good direction:
\(
w = \sum_{j \in \mathbb{S}} \eta_j k(x_j, \cdot),
\)
which achieves effective ``separation.'' Here, $\mathbb{S}$ is an index set identifying the support vectors selected by the SVM algorithm, and $\{\eta_j\}_{j \in \mathbb{S}}$ are the corresponding weights. 

The first algorithm aligns with the first perspective, aiming to separate the discrepancy between the nonparametric and parametric classes. The second algorithm incorporates the second perspective, aiming to separate the residuals from the origin point.

In practice, the parameter $\theta_0$ is unknown and must be estimated as $\hat{\theta}$. To mitigate estimation effects, we propose using a projection kernel $k_p(\cdot, \cdot)$, following the approach in \citealp{escanciano2024gaussian}.

To operationalize our approach, we introduce a $t$-type test statistic based on:
\[
    \mu_{\theta_0,\mathbb{S},k} = \mathbb{E}\left(\left \langle \varepsilon^\dagger_{\theta_0}k_p(X^\dagger,\cdot),\sum_{j \in \mathbb{S}} \eta_j k_p(x_j,\cdot)\right \rangle_{\mathcal{H}_k} \right) = \sum_{j \in \mathbb{S}}\eta_j \mathbb{E}\left[\varepsilon_{\theta_0}^\dagger k(X^\dagger, x_j)\right].
\]
Its empirical counterpart is given by
\[
    \hat \mu^\dagger_{\hat{\theta}, \mathbb{S},k_p} = \frac{1}{n} \sum_{i=1}^{n}  \left \langle \varepsilon^\dagger_{\hat{\theta},i} k_p(x_i^\dagger,\cdot), \sum_{j \in \mathbb{S}} \eta_j k_p(x_j,\cdot) \right \rangle_{\mathcal{H}_k}  = \frac{1}{n} \sum_{i=1}^{n} \sum_{j \in \mathbb{S}} \eta_j \varepsilon^\dagger_{\hat{\theta},i} k_p(x_i^\dagger, x_j),
\]
where variables with the dagger superscript ($\dagger$) denote test data, and those without denote training data. The sample size $n$ refers to the number of test observations. The expectation is taken with respect to the test data distribution.

This test statistic offers several advantages. First, it is computationally efficient, with linear time complexity, making it suitable for large datasets. Second, it is omnibus, meaning it is capable of detecting any arbitrary alternative hypothesis. Third, it admits a pivotal asymptotic distribution under the null hypothesis, which simplifies inference. To the best of our knowledge, no existing test statistic in the literature simultaneously satisfies all these criteria.

The remainder of the paper is organized as follows. Section \ref{sec:null_char} provides an equivalent characterization of the null hypothesis using $\mu_{\theta_0,\mathbb{S},k}$ and demonstrates its omnibus property, ensuring that the test can detect any deviation from the null. Section \ref{sec:test_stat} introduces the proposed test statistic and derives its asymptotic properties under the assumption that $\theta_0$ is known. Section \ref{sec:estimation_effect} addresses the challenges that arise when $\theta_0$ is unknown and must be estimated, proposing a correction through the use of a projection kernel $k_p(\cdot, \cdot)$. Section \ref{sec:ocsvm} formalizes the application of the Support Vector Machine (SVM) algorithm to learn the direction:
\(
w = \sum_{j \in \mathbb{S}} \eta_j k(x_j, \cdot),
\). Section \ref{sec:sim} presents simulation studies that validate the finite-sample performance of our method, along with real-data applications of the proposed approach. Finally, Section \ref{sec:conclusion} concludes the paper.

As introduced before, throughout the paper, we distinguish between training data and test data through the dagger superscript ($\dagger$): variables with the dagger superscript are testing points, while training points do not have any superscript. Both the training and testing points are sampled independently from the same distribution. Throughout the paper, we assume the following standard conditions hold: (i) the random variables $S=(Y,X)$ forms a strictly stationary process with probability measure $\mathbb{P}_{S}$; (ii) \textit{Reularity Conditions}. (1) the residual function $\varepsilon: \mathcal{S} \times \Theta \longrightarrow \mathbb{R}$ is continuous on $\Theta$ for each $s \in \mathcal{S}$; (2) $\mathbb{E}(\varepsilon(S;\theta)|X=x)$ exists and is finite for every $\theta \in \Theta$ and $x \in \mathcal{X}$ for which $\mathbb{P}_X(x) >0$; (3) $\mathbb{E}(\varepsilon(S;\theta)|X=x)$ is continuous on $\Theta$ for all $x \in \mathcal{X}$ for which $\mathbb{P}_X(x) >0$. We will write $\varepsilon(s_i;\theta)$ as $\varepsilon_{\theta,i}$ if there is no confusion.

\section{An Equivalent Statement of the Null and the Omnibus Property} \label{sec:null_char}
\subsection{An Equivalent Statement of the Null}
In this section, we demonstrate that $\mu_{\theta_0,\mathbb{S},k}^\dagger$, as defined in the Introduction, can be interpreted as a distance metric that quantifies the degree to which a parametric model fits the data. Throughout this section, we assume the availability of a training dataset $\{\varepsilon_{\theta_0,i} k(x_i,\cdot)\}_{i=1}^n$, and that the reproducing kernel $k(\cdot,\cdot)$ is integrally strictly positive definite:
\[
    \int_{\mathcal{X}} \int_{\mathcal{X}} f(x) k(x,x^\prime)  f(x^\prime) dP(x) dP(x^\prime) > 0, \quad \text{for any non-zero } f \in L^2(P).
\]
where \(L^2(P)\) denotes the Hilbert space of square-integrable functions with respect to the measure \(P\). Additionally, we assume that $\theta_0$ is known.

\begin{theorem}
    The null hypothesis $H_0$ holds almost surely if and only if $\mu_{\theta_0,\mathbb{S},k}^\dagger = 0$.
\end{theorem}
\begin{proof}
    See the Online Appendix A.
\end{proof}

As a running example throughout this paper, we adopt the widely used Gaussian kernel $k(x, x') = \exp\left(-\|x - x'\|_2^2 / \sigma\right)$, where $\sigma > 0$. For other commonly used reproducing kernels and their properties, see \cite{muandet2017kernel}.

\subsection{The Omnibus Property}

We now argue that the finite location mean difference metric $\mu_{\theta_0,\mathbb{S},k}^\dagger$ is omnibus, meaning it can detect any arbitrary alternative.

By Mercer's theorem, the reproducing kernel $k(x_j, x^\dagger)$ admits the decomposition:
\[
    k(x_j, x^\dagger) = \sum_{i=1}^\infty \lambda_i \phi_i(x_j) \phi_i(x^\dagger),
\]
where $\lambda_i$ and $\phi_i$ are the eigenvalues and eigenfunctions of the integral operator $T$, defined as:
\[
    Tf(x) = \int_{\Omega} k(x, x') f(x') dP(x'),
\]
with $P$ being a measure on the domain of $x'$, and $\Omega$ denoting the support. Mercer's theorem also ensures that the eigenfunctions $\phi_i(\cdot)$ form an orthonormal basis of $L^2(P)$, and the eigenvalues $\lambda_i$ are non-negative and decreasing. Consequently, we can express $\mu_{\theta_0,\mathbb{S},k}^\dagger$ as:
\[
    \mu_{\theta_0,\mathbb{S},k}^\dagger = \sum_{i=1}^\infty \left( \lambda_i \sum_{j \in \mathbb{S}} \eta_j \phi_i(x_j) \right) \mathbb{E}[\varepsilon_{\theta_0}^\dagger \phi_i(X^\dagger)].
\]
Equivalently, this can be written as:
\[
    \mu_{\theta_0,\mathbb{S},k}^\dagger = \sum_{i=1}^\infty \gamma_i \mathbb{E}[\varepsilon_{\theta_0}^\dagger \phi_i(X^\dagger)],
\]
where $\gamma_i = \lambda_i \sum_{j \in \mathbb{S}} \eta_j \phi_i(x_j)$.

The omnibus property of $\mu_{\theta_0,\mathbb{S},k}^\dagger$ arises from the fact that it is a linear combination of infinitely many orthonormal basis functions. This structure ensures that the metric captures all possible deviations from the null hypothesis, enabling it to detect any arbitrary alternative.

\section{Test Statistic and Its Asymptotic Properties}\label{sec:test_stat}

In this section, we study the asymptotic properties of the test statistic under the null hypothesis, fixed alternatives, and local alternatives, assuming $\theta_0$ is known and $\{\eta_j\}_{j \in \mathbb{S}}$ are given (have been learned from data points independent from the ones used in testing). In Section 4, we address the estimation effect when $\theta_0$ is replaced by a consistent estimator $\hat{\theta}$. In Section 5, we discuss the selection of $\{\eta_j\}_{j \in \mathbb{S}}$ via the SVM algorithms.

Our test statistic is based on the metric $\mu_{\theta_0,\mathbb{S},k}^\dagger$ defined in the previous section. Given $\{\eta_j\}_{j \in \mathbb{S}}$, we propose the following test statistic:
\[
	\hat{T}_{\theta_0,\mathbb{S},k} =  \frac{\hat{\mu}^\dagger_{\theta_0,\mathbb{S},k}}{\hat{\sigma}^\dagger_{\theta_0,\mathbb{S},k}},
\]
where 
\[
	\hat{\mu}^\dagger_{\theta,\mathbb{S},k} =\frac{1}{n} \sum_{j\in \mathbb{S}} \eta_j  \sum_{i=1}^n \varepsilon_{\theta,i}^\dagger k(x_i^\dagger,x_j) = \frac{1}{n} \sum_{j\in \mathbb{S}} \eta_j (\boldsymbol{\varepsilon}_{\theta}^\dagger)^\top \mathbf{K}(\mathbf{X}^\dagger,x_j),	
\]
where both $\boldsymbol{\varepsilon}_{\theta}^\dagger = (\varepsilon_{\theta,1}^\dagger,\ldots,\varepsilon_{\theta,n}^\dagger)$, and $\mathbf{K}(\mathbf{X}^\dagger,x_j) = (k(x_1^\dagger,x_j),\ldots,k(x_n^\dagger,x_j))^\top$ are $n \times 1$ vectors.

The empirical variance is 
\[
	\left(\hat{\sigma}^\dagger_{\theta,\mathbb{S},k}\right)^2 = \frac{1}{n-1} \sum_{i=1}^{n}\left(\sum_{j\in \mathbb{S}} \eta_j \varepsilon_{\theta,i} k(x_j,x_i^\dagger) - \hat{\mu}^\dagger_{\theta,\mathbb{S},k}\right)^2.
\]

\begin{theorem}
	Under the null hypothesis $H_0$, $\sqrt{n}\hat{T}_{\theta_0,J}$ converges in distribution to the standard normal distribution, and thus, $n\hat{T}^2_{\theta_0,J}$ converges in distribution to the $\chi^2_{1}$ distribution.
\end{theorem}

\begin{theorem}
	Under the fixed alternative hypothesis $H_1$, for any $t>0$, $\mathbb{P}(n \hat{T}^2_{\theta_0,J} > t) \to 1$.
\end{theorem}
The proofs for the above theorems can be obtained trivially by the Central Limit Theorem, Law of Large Numbers, and Slutsky's Theorem, hence omitted.
\begin{theorem}
	\label{thm-t-h1}
	Under the sequence of local alternatives $H_{1n}:\mathbb{E}(Y\mid X)=\mathcal{M}_{\theta_0}(X) + R(X)/\sqrt n$ with $\mathbb{E}|R(X)|< \infty$, we have
	\[
		 n \hat{T}^2_{\theta_0,\mathbb{S},j} \xrightarrow{d} \chi^2_{1}\left(\left(\frac{\sum_{j=1}^{J}\eta_j \Delta_j}{\sigma_{\theta_0,\mathbb{S},k}}\right)^2\right),
	\]
	where $\Delta_j = \mathbb{E}\left[R(X)k(X,x_j)\right]$ and $\chi^2_{1}(\lambda)$ is a non-central chi-square distribution with non-centrality parameter $\lambda$. 
\end{theorem}
\begin{proof}
	See the Online Appendix A.
\end{proof}

\section{Dealing with Estimation Effects}\label{sec:estimation_effect}
So far, we have assumed that the value of $\theta_0$ is known. In practice, $\theta_0$ is estimated by a consistent estimator $\hat{\theta}$. In this section, we discuss how to deal with the estimation effect when $\theta_0$ is estimated. 

\subsection{Eliminating Estimating Effects via Projection}
The following assumptions are required for this purpose. 
\begin{assumption}
	(i) The parameter space $\Theta$ is a compact subset of $\mathbb{R}^q$; (ii) The true parameter $\theta_0$ is an interior point of $\Theta$; and (iii) The consistent estimator $\hat{\theta}$ satisfies $\lVert \hat{\theta} - \theta_0 \rVert = O_p (n^{-\alpha})$, with $\alpha > 1/4$.
\end{assumption}

\begin{assumption}
	(i) The residual $\varepsilon_{\theta}$ is twice continuously differentiable with respect to $\theta$, with its first derivative $g_{\theta}(x) = \mathbb{E}(\nabla_{\theta} \varepsilon_{\theta}|X=x)$ satisfying $\mathbb{E}\left(\sup_{\theta \in \Theta} \lVert g_{\theta}(X) \rVert\right) < \infty$ and its second derivative satisfying $\mathbb{E}\left(\sup_{\theta \in \Theta} \lVert \nabla g_{\theta}(X) \rVert\right) < \infty$; (ii) the matrix $\Gamma_{\theta} = \mathbb{E}\left[g_{\theta}(X) g_{\theta}(X) ^\top\right]$ is nonsigular in a neighborhood of $\theta_0$.
\end{assumption}

Assumption 1 is weaker than the related conditions in the literature. We only impose that $\hat{\theta}$ converges in probability at a slower rate than usual. Additionally, we do not require it to admit an asymptotically linear representation. This could be useful in the context of non-standard estimation procedures, such as the LASSO. Assumption 2 is standard in the literature and imposes regularity conditions on the smoothness of the residual function.

We now introduce a projection operator $\boldsymbol{\Pi}: \mathcal{H}_k \longrightarrow \mathcal{H}_k$ defined as:
\[
    (\boldsymbol{\Pi}\omega)(x) = \omega(x) - \mathbb{E}\left[\omega(X) (g_{\theta_0}(X))^\top \right] \Gamma_{\theta_0}^{-1} g_{\theta_0}(x), \forall \omega \in \mathcal{H}_k, x \in \mathcal{X},
\]
Applying this projection operator to a kernel function $k(x,x^\dagger)$ yields the projection kernel:
\[
    k_p(x,x^\dagger)=\boldsymbol{\Pi}k(x,x^\dagger) = k(x,x^\dagger) - \mathbb{E}\left[k(x,X^\dagger) (g_{\theta_0}(X^\dagger))^\top\right] \Gamma_{\theta_0}^{-1} g_{\theta_0}(x^\dagger).
\]

To analyze the local behavior of the projected mean embedding $\mathbb{E}(\varepsilon_{\theta_0}k_p(x,X^\dagger))$ in a neighborhood of $\theta_0$, consider their derivatives with respect to $\theta$ evaluated at $\theta_0$:
\begin{align*}
    \frac{\partial}{\partial \theta} \mathbb{E}(\varepsilon_{\theta}k_p(x,X^\dagger))\bigg|_{\theta = \theta_0} &= \mathbb{E}\left(  k(x,X^\dagger)(g_{\theta_0}(X^\dagger))^\top - \mathbb{E}\left[k(x,X^\dagger) (g_{\theta_0}(X^\dagger))^\top\right] \Gamma_{\theta_0}^{-1} g_{\theta_0}(X^\dagger)(g_{\theta_0}(X^\dagger))^\top\right) \\
    &= \boldsymbol{0}.
\end{align*}
This establishes that the projected mean embedding are locally robust to small perturbations in $\theta$ around $\theta_0$. The vanishing derivatives follow from the dominated convergence theorem.

The matrix estimator (using the testing data) to this projection operator is given by:
\[
    \hat{\boldsymbol{\Pi}}^\dagger = \boldsymbol{I}_n -  \hat{\mathbf{g}}\left(\hat{\mathbf{g}}^\top \hat{\mathbf{g}}\right)^{-1} \hat{\mathbf{g}}^\top
\]
where $\hat{\mathbf{g}}$ is a $n \times d$ matrix of scores whose $i$th row is given by $(\hat g_{i}^\dagger)^\top = (\nabla_{\theta} \varepsilon_{\theta}^\dagger|_{\theta = \hat \theta})^\top$, and $\boldsymbol{I}_n$ is the $n \times n$ identity matrix.  The projected version of the kernel vector ($\boldsymbol{K}(\boldsymbol{X}^\dagger,x_j)$) is given by:
\[
    \boldsymbol{K}_p(\boldsymbol{X}^\dagger,x_j) = (\hat{\boldsymbol{\Pi}}^\dagger)^\top \boldsymbol{K}(\boldsymbol{X}^\dagger,x_j) 
\]

The following theorem states how this projection kernel eliminates the estimation effect when vector multiplication is performed.
\begin{theorem}
	Suppose Assumption 1 holds, then 
	\[
		\frac{1}{n}(\boldsymbol{\hat \varepsilon}^\dagger)^\top \mathbf{K}_p(\mathbf{X}^\dagger,\cdot) =  \frac{1}{n}(\boldsymbol{\hat \varepsilon}_{p}^\dagger)^\top \mathbf{K}(\mathbf{X}^\dagger,\cdot)  = \frac{1}{n}\left(\boldsymbol{\varepsilon}_{p,\theta_0}^\dagger\right)^\top \mathbf{K}(\mathbf{X}^\dagger,\cdot) + O_p(n^{-2\alpha}),
	\]
	where  
	\[
		\boldsymbol{\hat \varepsilon}_{p}^\dagger =  \hat{\boldsymbol{\Pi}}^\dagger\boldsymbol{\hat \varepsilon}^\dagger,
	\]
	and 
	\[
		\boldsymbol{\varepsilon}_{p,\theta_0}^\dagger =  \hat{\boldsymbol{\Pi}}^\dagger \boldsymbol{\varepsilon}_{\theta_0}^\dagger.
	\]
\end{theorem}
\begin{proof}
	See the Online Appendix A.
\end{proof}

As long as the convergence speed satisfies $\alpha>1/4$, we have for any weight $\eta_j$,
\[
	\begin{split}
		\sqrt{n}\hat{\mu}^\dagger_{\hat \theta,\mathbb{S},k_p} & = \frac{1}{\sqrt{n}} \sum_{j \in \mathbb{S}} \eta_j (\boldsymbol{\hat \varepsilon}^\dagger)^\top \mathbf{K}_p(\mathbf{X}^\dagger,x_j) \\
		& = \frac{1}{\sqrt{n}} \sum_{j \in \mathbb{S}} \eta_j (\boldsymbol{\hat \varepsilon}_{p}^\dagger)^\top \mathbf{K}(\mathbf{X}^\dagger,x_j) \\
		& = \frac{1}{\sqrt{n}} \sum_{j \in \mathbb{S}} \eta_j (\boldsymbol{\varepsilon}_{p, \theta_0}^\dagger)^\top \mathbf{K}(\mathbf{X}^\dagger,x_j) + o_p(1) \\ 
	\end{split}
\]
The corresponding empirical variance and test statistic are 
\[
	\left(\hat{\sigma}^\dagger_{\hat \theta,\mathbb{S},k_p} \right)^2= \frac{1}{n-1} \sum_{i=1}^{n}\left(\sum_{j=1}^{J}\eta_j \varepsilon_{p,\hat\theta,i}^\dagger k(x_j,x_i^\dagger) - \hat{\mu}_{\hat\theta,\mathbb{S},k_p}\right)^2,
\]

\[
	 \hat T_{\hat \theta, \mathbb{S},k_p} = \frac{\hat{\mu}^\dagger_{\hat \theta,\mathbb{S},k_p}}{\hat{\sigma}^\dagger_{\hat \theta,\mathbb{S},k_p}}.
\]

The asymptotic results for the $t$-statistic (Theorems 2-4) hold under the projection kernel $k_p(\cdot,\cdot)$ free from the estimation effect.

\subsection{Bootstrap-based Critical Values}
For our theoretical analysis, we use asymptotic critical values. However, since the weights $\{\eta_j\}_{j \in \mathbb{S}}$ and location points $\{x_j\}_{j \in \mathbb{S}}$ are chosen in a data-dependent manner, we generally recommend obtaining critical values via a multiplier bootstrap procedure to ensure proper control of the type I error at finite sample sizes, particularly when the sample size is relatively small.

For simplicity and ease of implementation, we compute the $t$-statistic without normalization and use $\hat{\mu}_{p, \hat{\theta}, J}$ instead of $\hat T_{p,\hat \theta, J}$ to construct test statistic. A detailed bootstrap procedure is provided in the Online Appendix B.

We approximate the asymptotic null distribution of $\sqrt{n}\hat{\mu}_{\hat \theta,\mathbb{S},k_p}^\dagger$ by that of $\sqrt{n}\hat{\mu}_{\hat{\theta}, \mathbb{S},k_p}^*$, where 
\[
	\sqrt{n}\hat{\mu}_{\hat{\theta}, \mathbb{S},k_p}^* = \frac{1}{\sqrt{n}} \sum_{j \in \mathbb{S}} \eta_j (\boldsymbol{\hat \varepsilon}_{p}^*)^\top \mathbf{K}(\mathbf{X}^\dagger,x_j)
\]
and 
\[
	\boldsymbol{\hat \varepsilon}_{p}^* = \hat{\boldsymbol{\Pi}}^\dagger (\boldsymbol{\hat \varepsilon}^\dagger \odot \mathbf{V}) .
\]
with $\mathbf{a} \odot \mathbf{b}$ being element-wise multiplication (Hadamard product) of vectors of the same size. $\mathbf{V}$ is a random vector of size $n$ with i.i.d random variables satisfying $\mathbb{E}(v_1) = 0$ and $\mathrm{Var}(v_1) = 1$. Notable examples include the Rademacher, standard normal, and Bernoulli random variables with 
\[
	P(v_1 = 0.5(1-\sqrt{5})) = b, \quad P(v_1 = 0.5(1+\sqrt{5})) =1- b
\] 
where $b = (1+\sqrt{5})/2\sqrt{5}$, see \cite{mammen1993bootstrap}.

To theoretically justify the bootstrap approximation, no additional assumptions are needed. In contrast, other related bootstrap procedures often require extra conditions, such as restrictions on the bootstrap version of the estimator. These conditions may not hold in certain cases, such as when using the LASSO method.

\begin{theorem}
    Under Assumptions 1 and 2, if $|v_1| < c$ with probability 1 for some finite constant $c$, $\mathbb{E}(v_1) = 0$, and $\mathrm{Var}(v_1) = 1$, then 
    \[
        \sup_t \left| P\left(\sqrt{n}\hat{\mu}_{\hat{\theta}, \mathbb{S},k_p}^* < t\right) - P\left(c_{\infty} < t\right)\right| = o_p(1),
    \]
    where $c_{\infty}$ follows a normal distribution with mean zero and variance
    \[
        \sigma^2_{\theta_0, \mathbb{S},k_p} = \mathrm{Var}\left(\sum_{j=1}^{J}\eta_j k(x_j,X^\dagger)\varepsilon^\dagger_{p,\theta_0}\right).
    \]
\end{theorem}
\begin{proof}
    See the Online Appendix A. 
\end{proof}

Finally, we highlight that the time complexity of the proposed bootstrap procedure is $O(B n)$, where $B$ is the bootstrap size and $n$ is the test sample size. In contrast, for other KCM-type tests, the time complexity is $O(B n^2)$.

\section{Determine $\mathbb{S}$ and $\{\eta_j\}$ via Support Vector Machines}\label{sec:ocsvm}

Theorems 2 and 3 establish that the test power increases with $\left|\hat{T}_{\hat{\theta}, \mathbb{S}, k_p}\right|$. A natural approach to enhancing test power is to maximize this signal-to-noise ratio. In the context of classification, this objective aligns with Fisher's discriminant analysis, which seeks to maximize the between-class variance while minimizing the within-class variance.

However, in an infinite-dimensional space such as a RKHS, directly inverting the covariance operator is ill-posed. Although Tikhonov regularization can stabilize the inversion process, it is computationally expensive.

In this section, we advocate and formalize the application of two Support Vector Machine (SVM) algorithms to learn the direction:
\(
w = \sum_{j \in \mathbb{S}} \eta_j k(x_j, \cdot).
\)
The first algorithm is the classical $\nu$-SVM, which aims to find a hyperplane orthogonal to $w$ that maximizes the margin between the nonparametric and parametric classes. The second algorithm treats the residual points (differences between the nonparametric and parametric parts) as a one-class classification problem, aiming to maximize the distance between the origin and a hyperplane that encloses the residual points. We demonstrate that these two algorithms increase lower bounds of the signal-to-noise ratio, thereby improving the power of the test.

We provide comprehensive details on the test statistic construction algorithms in the Online Appendix B.

\subsection{$\nu-$SVM Algorithm in Specification Testing}

The $\nu$-SVM is a supervised learning algorithm designed for class separation. It learns a hyperplane that maximizes the distance to the nearest training data point of any class, effectively maximizing the margin between classes.

In our framework, the hyperplane is defined as:
\[
    \left\{z_p k(x, \cdot) = (y_{p} k(x, \cdot), \mathcal{M}_{\hat{\theta},p}(x)k(x,\cdot)) \in \mathcal{H}_k : \left\langle w, z_p \right\rangle_{\mathcal{H}_k} + b = 0 \right\},
\]
where $w$ is the hyperplane's normal vector and $b$ is the offset parameter. Here, $y_p$ and $\mathcal{M}_{\hat{\theta},p}(x)$ are generated using the same projection procedure as $\hat{\varepsilon}_p$. The dataset contains $2n$ training points, with half belonging to the nonparametric class and the other half to the parametric class.

Standard SVM typically assumes that data points lie in the same orthant, often the positive orthant, in the feature space. This assumption is often satisfied using a Gaussian kernel, which ensures that $\{k(x_i,\cdot)\}_{i=1}^n$ reside in the same orthant with unit length (see Section 2.3 of \cite{scholkopf1999support}).

However, in our setting, the training data points $\{z_{p,i} k(x_i, \cdot)\}_{i=1}^{2n}$ do not lie in the same orthant, even when using a Gaussian kernel. To address this issue, we introduce shifted data points:
\[
\{\tilde{z}_{p,i} k(x_i, \cdot)\}, \quad \text{where} \quad \tilde{z}_{p,i} = \{y_{p,i}, \mathcal{M}_{\hat{\theta},p,i}\} + e > 0 \; \forall i,
\]
with $e > 0$ being a constant shift. This transformation does not affect the properties of the test statistic because the shift is applied uniformly to both classes, and the test statistic is based on the difference between the two classes:
\[
\hat{\varepsilon}_{p} k(x,\cdot) = (y_p + e - \mathcal{M}_{\hat{\theta},p}(x) - e) k(x,\cdot).
\]

We then use the shifted data points to select $w$ via SVM. The primal optimization problem is given by
\[
    \max_{w, b, \xi, \rho} \quad \nu \rho - \frac{1}{n} \sum_{i=1}^n \xi_i - \frac{1}{2} \|w\|^2,
\]
subject to
\[
    l_i (\langle w, \tilde{z}_{p,i} k(x,\cdot)\rangle_{\mathcal{H}_k} + b) \geq \rho - \xi_i, \quad \forall i = 1, \dots, 2n,
\]
\[
    \xi_i \geq 0, \quad \forall i = 1, \dots, 2n, \quad \rho \geq 0.
\]
Here, $\{l_i\}_{i=1}^{2n}$ are label variables that take the value $1$ for the nonparametric class and $-1$ for the parametric class. The parameter $\nu$ is a hyperparameter that balances the trade-off between maximizing the margin ($\rho$) and minimizing the number of support vectors. The slack variables $\{\xi_i\}_{i=1}^{2n}$ are introduced to allow for misclassification in the training data, thereby avoiding overfitting.

The margin parameter $\rho$ can be interpreted as the deviation signal, while $\|w\|_{\mathcal{H}_k}$ measures the noise level in our context. Specifically, for each class, if the data points are correctly classified by the hyperplane, their projections onto the direction $w$ should lie at least $\rho / \|w\|_{\mathcal{H}_k}$ away from the separating hyperplane. Consequently, the distance between the projected sample averages of the nonparametric and parametric parts should be at least $2 \rho / \|w\|_{\mathcal{H}_k}$, representing a signal-to-noise ratio that the SVM algorithm aims to maximize.

Solving the primal problem is challenging due to the infinite-dimensional nature of the RKHS. However, leveraging duality theory, we reformulate the primal problem in terms of dual variables $\{\alpha_j\}_{j=1}^n$. The resulting dual problem is
\[
\min_{\alpha} \frac{1}{2} \sum_{i,j=1}^{2n} \alpha_i \alpha_j \, \tilde{z}_{p,i} \tilde{z}_{p,j} \, l_i l_j \, k(x_i,x_j),
\]
subject to
\[
0 \leq \alpha_i \leq 1/(2n), \quad \sum_{i=1}^{2n} \alpha_i \tilde{z}_{p,i} = 0, \quad \text{and} \quad \sum_{i=1}^{2n} \alpha_i \geq \nu,
\]
where $\{x_i = x_{n+i}\}_{i=1}^n$.

It can be shown that the training points associated with $\alpha_j > 0$ are support vectors, while the remaining points have $\alpha_j = 0$. Consequently, the index set $\mathbb{S}$ is determined by the indices of the support vectors. The weights $\{\eta_j\}_{j \in \mathbb{S}}$ are then given by:
\[
\eta_j = \alpha_j \tilde{z}_{p,j} l_j.
\]

\subsection{One-Class SVM in Specification Testing}
Another approach to selecting $\mathbb{S}$ and $\{\eta_j\}$ is to use the One-Class Support Vector Machine (OCSVM) algorithm. OCSVM is an unsupervised algorithm designed for anomaly detection. It learns a hyperplane that separates the majority of data points, considered normal, from the origin, identifying points near the origin as anomalies. Geometrically, OCSVM maximizes the margin between this hyperplane and the origin while minimizing the number of data points within the margin.

In our context, we treat the differences between the nonparametric and parametric parts, i.e., $\{\hat{\varepsilon}_{p,i} k(x_i,\cdot)\}_{i=1}^n$, as `normal' data points, while the origin represents the null hypothesis and is considered an anomaly.

Similar to the two-class SVM, we need to transform the data points to lie in the same orthant:
\[
\{\tilde{\varepsilon}_{p,i} k(x_i, \cdot)\}, \quad \text{where} \quad \tilde{\varepsilon}_{p,i} = \hat{\varepsilon}_{p,i} + e > 0 \; \forall i,
\]
with $e > 0$ being a constant shift. This transformation ensures all data points are positive and avoids issues caused by mixed signs.

We then use the shifted data points $\{\tilde{\varepsilon}_{p,i} k(x_i, \cdot)\}_{i=1}^n$ to select $w$ via OCSVM. The hyperplane is characterized as:
\[
    \left\{\tilde{\varepsilon}_{p} k(x, \cdot) \in \mathcal{H}_k : \left\langle w, \tilde{\varepsilon}_{p} k(x, \cdot) \right\rangle_{\mathcal{H}_k} = \rho \right\}.
\]

To separate the majority of data points from the origin, the parameter $\rho$, which serves as a signal of deviation from the null hypothesis, should be maximized, while the norm of $w$ should be minimized. This objective is achieved by solving the following OCSVM optimization problem:
\[
    \begin{split}
        \max_{w,\xi,\rho} & \quad \rho - \frac{1}{\nu n} \sum_{i=1}^n \xi_i - \frac{1}{2} \|w\|_{\mathcal{H}_k}^2, \\
        \text{s.t.} & \quad \langle w, \tilde{\varepsilon}_{p,i} k(x_i,\cdot) \rangle_{\mathcal{H}_k} \geq \rho - \xi_i, \quad i=1,\ldots,n, \\
        & \quad \xi_i \geq 0, \quad i=1,\ldots,n.
    \end{split}
\]

The dual problem of OCSVM is given by
\[
\min_{\alpha} \frac{1}{2} \sum_{i,j=1}^n \alpha_i \alpha_j \, \tilde{\varepsilon}_{p,i} \, k(x_i,x_j) \, \tilde{\varepsilon}_{p,j},
\]
subject to
\[
0 \leq \alpha_i \leq 1/(\nu n), \quad \sum_{i=1}^n \alpha_i = 1, \quad \forall i=1,\ldots,n.
\]

As with the two-class SVM, the training points associated with $\alpha_j > 0$ are support vectors, while the remaining points have $\alpha_j = 0$. Consequently, $\mathbb{S}$ is determined by the indices of the support vectors. The weights $\{\eta_j\}_{j \in \mathbb{S}}$ are then given by:
\[
\eta_j = \alpha_j \tilde{\varepsilon}_{p,j}.
\]


\subsection{Power Analysis of SVM-based Test Statistics}

It is clear that the test power is determined by the value of
\[
\left|T_{\hat{\theta}, \mathbb{S}, k_p}\right| = \left| \frac{\mu_{\hat{\theta}, \mathbb{S}, k_p}}{\sigma_{\hat{\theta}, \mathbb{S}, k_p}} \right|,
\]
where
\[
\begin{split}
    & \mu_{\hat{\theta}, \mathbb{S}, k_p} = \mathbb{E}\left(\langle \varepsilon^\dagger_{\hat{\theta}, p} k(X^\dagger, \cdot), w \rangle_{\mathcal{H}_k}\right), \\
    & \sigma^2_{\hat{\theta}, \mathbb{S}, k_p} = \mathrm{Var}\left(\langle \varepsilon^\dagger_{\hat{\theta}, p} k(X^\dagger, \cdot), w \rangle_{\mathcal{H}_k}\right).
\end{split}
\]

In this subsection, we investigate how the generalization error of the SVM algorithms affects the test power.

We begin with the two-class SVM algorithm and its generalization margin error bound (Theorem 7.3 in \cite{scholkopf2002learning}). This result states that with probability at least $1-\delta$, we have:
\[
\mathrm{Pr}\left(l^\dagger \langle \tilde{z}^\dagger_p k(x^\dagger, \cdot), w^* \rangle_{\mathcal{H}_k} < \rho \right) \leq \underbrace{\omega + \sqrt{\frac{c}{n} \left( \frac{R^2 \Lambda^2}{\rho^2} \ln^2(n) + \ln(1/\delta) \right)}}_{d},
\]
where $w^*$ denotes the solution of the SVM, characterized by the index set $\mathbb{S}^*$ of support vectors and corresponding dual coefficients. Here $c$ is a universal constant, $\omega$ is the fraction of training points with margin smaller than $\rho / \|w^*\|_{\mathcal{H}_k}$, $R$ and $\Lambda$ are bounds on the norm of the data points and the norm of the weight vector, respectively: $\| z^\dagger k(x^\dagger, \cdot)\|_{\mathcal{H}_k} \leq R$ and $\| w^*\|_{\mathcal{H}_k} \leq \Lambda$, and $n$ is the sample size of the training dataset.

This generalization error bound implies that for a new test point, the probability of misclassifications by the learned separating hyperplane is bounded by $d$, which decreases as the margin parameter $\rho$ increases.

Without loss of generality, assume that the nonparametric part, after projecting onto $w^*$, has positive values for correctly classified points, while the parametric part has negative values for correctly classified points. Then, with probability at least $1-\delta$, we have:
\begin{align*}
    & \mathrm{Pr} \left( \langle \tilde{Y}_p^\dagger k(x^\dagger, \cdot), w^* \rangle_{\mathcal{H}_k} > \rho \right) > 1-d, \\
    & \mathrm{Pr} \left( \langle \widetilde{\mathcal{M}}_{\hat{\theta}, p}^\dagger(x^\dagger) k(x^\dagger, \cdot), w^* \rangle_{\mathcal{H}_k} < -\rho \right) > 1-d.
\end{align*}

Let $R_w$ be the bound on the projected value of the data points, i.e., $|\langle \tilde{z}^\dagger_p k(x^\dagger, \cdot), w^* \rangle_{\mathcal{H}_k}| \leq R_w$. It follows that:
\[
\langle \tilde{Y}_p^\dagger k(x^\dagger, \cdot), w^* \rangle_{\mathcal{H}_k} \in
\begin{cases}
[\rho, R_w], & \text{if correctly classified}, \\
[-R_w, \rho), & \text{if misclassified},
\end{cases}
\]
and
\[
\langle \widetilde{\mathcal{M}}_{\hat{\theta}, p}^\dagger(x^\dagger) k(x^\dagger, \cdot), w^* \rangle_{\mathcal{H}_k} \in
\begin{cases}
[-R_w, -\rho], & \text{if correctly classified}, \\
(-\rho, R_w], & \text{if misclassified}.
\end{cases}
\]

Subsequently, we can bound $\left|T_{\hat{\theta}, \mathbb{S}^*, k_p}\right|$ as:
\[
\left|T_{\hat{\theta}, \mathbb{S}^*, k_p}\right| = \frac{\mu_{\hat{\theta}, \mathbb{S}^*, k_p}}{\sigma_{\hat{\theta}, \mathbb{S}^*, k_p}} \geq \frac{(1-d)^2 2\rho}{\sigma_{\hat{\theta}, \mathbb{S}^*, k_p}} + \frac{2(\rho - R_w)(1-d)d}{\sigma_{\hat{\theta}, \mathbb{S}^*, k_p}} - \frac{2R_w d^2}{\sigma_{\hat{\theta}, \mathbb{S}^*, k_p}}.
\]

The last two terms are negative and negligible if $d$ is small. The first term is positive by construction and increases with $\rho$. Furthermore, the standard deviation $\sigma_{\hat{\theta}, \mathbb{S}^*, k_p}$ is positively related to $\|w^*\|_{\mathcal{H}_k}$. Combining these observations, we conclude that the test power is positively related to the margin distance $\rho / \|w^*\|_{\mathcal{H}_k}$, which is maximized by the SVM algorithm.

The generalization error bound for OCSVM (Theorem 8.6 in \cite{scholkopf2002learning}) is more complex than that of two-class SVM. However, one definitive conclusion can be drawn: the distance between the origin and the hyperplane, $\rho / \|w^*\|_{\mathcal{H}_k}$, is inversely related to the probability bound $d$ of the generalization error.

To investigate how the OCSVM generalization error affects the test power, assume without loss of generality that $\mu_{\hat{\theta}, \mathbb{S}^*, k_p} < 0$. Then, with probability at least $1-\delta$, we have:
\[
\mathrm{Pr}\left(\langle \tilde{\varepsilon}^\dagger_p k(x^\dagger, \cdot), w^* \rangle_{\mathcal{H}_k} < \rho\right) = F\left(\frac{\rho + \mu_{\hat{\theta}, \mathbb{S}^*, k_p}}{\sigma_{\hat{\theta}, \mathbb{S}^*, k_p}}\right) \leq d,
\]
where $F$ is the cumulative distribution function of the random variable:
\[
\frac{\langle \tilde{\varepsilon}^\dagger_p k(x^\dagger, \cdot), w^* \rangle_{\mathcal{H}_k} + \mu_{\hat{\theta}, \mathbb{S}^*, k_p}}{\sigma_{\hat{\theta}, \mathbb{S}^*, k_p}}.
\]

By the properties of $F$, we obtain:
\[
\begin{split}
    & \frac{\rho + \mu_{\hat{\theta}, \mathbb{S}^*, k_p}}{\sigma_{\hat{\theta}, \mathbb{S}^*, k_p}} \leq F^{-1}(d), \\
    & \left|T_{\hat{\theta}, \mathbb{S}^*, k_p}\right| = -\frac{\mu_{\hat{\theta}, \mathbb{S}^*, k_p}}{\sigma_{\hat{\theta}, \mathbb{S}^*, k_p}} \geq \frac{\rho}{\sigma_{\hat{\theta}, \mathbb{S}^*, k_p}} - F^{-1}(d).
\end{split}
\]

As $\rho / \|w^*\|_{\mathcal{H}_k}$ increases (and consequently $d$ decreases), the second term $-F^{-1}(d)$ also increases. The first term is positively related to the margin distance $\rho / \|w^*\|_{\mathcal{H}_k}$. Therefore, the test power is positively related to the margin distance $\rho / \|w^*\|_{\mathcal{H}_k}$, which is maximized by the OCSVM algorithm.

The above analysis also suggests that artificially shifting the data points to have large values of $\tilde{\varepsilon}_p$ would have little effect on the test power. This is because such shifting also moves the location of $F(\cdot)$ to the right. Hence, holding everything else constant, although $\rho$ becomes larger, so does $F^{-1}(d)$. Consequently, the increase in $\rho$ is counterbalanced by the corresponding increase in $F^{-1}(d)$, leaving the overall test power largely unaffected.

\section{Simulation Evidences and Empirical Studies}\label{sec:sim}

\subsection{Simulation Evidences}
We consider the following simulation designs. The null model is given by
\[
	DGP_1: Y = \theta_0^\top X + \varepsilon.
\]
The alternative models are given by
\[
	\begin{split}
		& DGP_2: Y = \theta_0^\top X + c \exp(-(\theta_0^\top X)^2) + \varepsilon, \\
		& DGP_3: Y = \theta_0^\top X + 3c \cos(0.6\pi \theta_0^\top X) + \varepsilon, \\
		& DGP_4: Y = \theta_0^\top X + 0.5c (\theta_0^\top X)^2 + \varepsilon, \\
		& DGP_5: Y = \theta_0^\top X + 0.5c \exp(0.25\theta_0^\top X) + \varepsilon,
	\end{split}
\]
where $\theta_0$ are $q$-dimensional vectors with first $p$ and the last $q-p$ elements being 1 and 0, respectively. Here, we set $c=0.25$, $p = \lfloor 0.1 q \rfloor$ and $q = \{10,20\}$. The covariates $X$ follow a $q$-dimensional standard normal distribution. The error term $\varepsilon$ is a standard normal random variable. Similar DGPs have been considered in \cite{escanciano2006consistent,tan2022integrated,escanciano2024gaussian}. 

Beside the proposed SVM based t-type statistics ($\hat T_{\nu-SVM}$ and $\hat{T}_{OCSVM}$), we consider three alternative test statistics: the ICM test ($\hat{T}_{ICM}$) by \cite{bierens1982consistent}, the KCM test ($\hat{T}_{KCM}$) by \cite{muandet2020kernel}, and the Gaussian Process test ($\hat{T}_{GP}$) of \cite{escanciano2024gaussian}. All three alternative statistics take the form of 
\[
	n \hat{T} = \frac{1}{n}\sum_{i,j = 1}^{n} \varepsilon_{\hat \theta,i} k(x_i,x_j)\varepsilon_{\hat \theta,j}.
\] 
All tests employ the Gaussian kernel $k(x, y) = \exp(-\|x - y\|_2^2 / \sigma)$. For the ICM test, we fix $\sigma = 2$. For the $\nu$-SVM, OCSVM, KCM, and GP tests, we follow \cite{escanciano2024gaussian} and select $\sigma$ using the median heuristic, i.e., $\sigma = \text{median}(\{\|x_i - x_j\|_2\}_{i \neq j})$. Both the proposed SVM-based tests and the GP test utilize the projection method described in Section 4 to mitigate estimation effects. For the KCM and GP tests, critical values are obtained via a multiplier bootstrap procedure as outlined in Section 4, while the ICM test uses the wild bootstrap procedure of \cite{delgado2006consistent}. For the SVM-based tests, we report results using both analytic and multiplier bootstrap critical values. The number of bootstrap replications is set to $B = 500$, and each simulation scenario is repeated $R = 1000$ times. Empirical sizes and powers are computed as the proportion of rejections over the $R$ replications.

For the SVM-based methods, we allocate $10\%$ of the data for training the one-class SVM and use the remaining $90\%$ for testing. SVM implementations are carried out using the \texttt{scikit-learn} Python package (\texttt{OneClassSVM} for OCSVM and \texttt{NuSVC} for $\nu$-SVM). Due to space constraints, additional simulation designs and results are provided in Online Appendix C.

Regarding the simulation results, we observe that the performance of the SVM-based test statistics is similar under analytic and multiplier bootstrapped critical values. However, it should be emphasized that, in general, the finite-sample performance using bootstrap critical values is superior to that using analytic critical values. This is demonstrated by the additional simulation studies documented in Online Appendix C.

Both SVM-based test statistics demonstrate strong finite-sample performance compared to the other tests (Tables \ref{tab:q10-5} and \ref{tab:q20-5}). When the covariate dimension is moderate ($q = 10$), the SVM-based and GP tests exhibit accurate size control, whereas the ICM and KCM tests suffer from size distortion. In terms of power, the SVM-based tests perform comparably to the GP test, while the ICM and KCM tests show lower power. When the covariate dimension is high ($q = 20$), SVM-based tests maintain accurate size control and high power, whereas the other tests experience significant size distortion and reduced power.

\begin{table}[htbp]
    \centering
    \footnotesize
    \caption{Empirical sizes and powers at $5\%$ estimated by OLS with $q=10$}
    \begin{adjustbox}{width=\textwidth}
        \begin{tabular}{@{}lccccccccccc@{}}
            \toprule
            \multicolumn{1}{l}{$n$} & \multicolumn{5}{c}{200} & \multicolumn{6}{c}{400} \\
            \cmidrule(lr){2-6} \cmidrule(lr){8-12}
            & $\hat{T}_{\nu\text{-SVM}}$ & $\hat{T}_{OCSVM}$ & $\hat{T}_{GP}$ & $\hat{T}_{KCM}$ & $\hat{T}_{ICM}$ 
            & & $\hat{T}_{\nu\text{-SVM}}$ & $\hat{T}_{OCSVM}$ & $\hat{T}_{GP}$ & $\hat{T}_{KCM}$ & $\hat{T}_{ICM}$ \\
            
            \\
            \multicolumn{12}{c}{SIZE} \\
            \midrule
            $DGP_{1}$ (Bootstrap) & 0.066 & 0.058 & 0.030 & 0.001 & 0.000 && 0.046 & 0.055 & 0.047 & 0.002 & 0.000  \\
            \qquad \quad \,(Analytic) & [0.055] & [0.055] & - & - & - && [0.053] & [0.049] & - & - & -  \\
            
            \\
            \multicolumn{12}{c}{POWER} \\
            \midrule
            $DGP_{2}$ (Bootstrap) & 0.381 & 0.436 & 0.362 & 0.086 & 0.005 && 0.689 & 0.730 & 0.699 & 0.369 & 0.044  \\ 
            \qquad \quad \,(Analytic) & [0.403] & [0.373] & - & - & - && [0.723] & [0.687] & - & - & -  \\
            \\ 
            $DGP_{3}$ (Bootstrap) & 0.439 & 0.433 & 0.633 & 0.170 & 0.045 && 0.720 & 0.757 & 0.956 & 0.711 & 0.422  \\
            \qquad \quad \,(Analytic) & [0.430] & [0.436] & - & - & - && [0.739] & [0.701] & - & - & -  \\
            \\ 
            $DGP_{4}$ (Bootstrap) & 0.168 & 0.171 & 0.162 & 0.015 & 0.001 && 0.367 & 0.346 & 0.350 & 0.091 & 0.002  \\
            \qquad \quad \,(Analytic) & [0.181] & [0.181] & - & - & - && [0.358] & [0.335] & - & - & -  \\
            \\ 
            $DGP_{5}$ (Bootstrap) & 0.305 & 0.303 & 0.263 & 0.042 & 0.001 && 0.573 & 0.526 & 0.509 & 0.221 & 0.021  \\
            \qquad \quad \,(Analytic) & [0.289] & [0.300] & - & - & - && [0.554] & [0.544] & - & - & -  \\
            \bottomrule
        \end{tabular}
    \label{tab:q10-5}
    \end{adjustbox}
\end{table}

\begin{table}[htbp]
    \centering
    \footnotesize
    \caption{Empirical sizes and powers at $5\%$ estimated by OLS with $q=20$}
    \begin{adjustbox}{width=\textwidth}
        \begin{tabular}{@{}lccccccccccc@{}}
            \toprule
            \multicolumn{1}{l}{$n$} & \multicolumn{5}{c}{200} & \multicolumn{6}{c}{400} \\
            \cmidrule(lr){2-6} \cmidrule(lr){8-12}
            & $\hat{T}_{\nu\text{-SVM}}$ & $\hat{T}_{OCSVM}$ & $\hat{T}_{GP}$ & $\hat{T}_{KCM}$ & $\hat{T}_{ICM}$ 
            & & $\hat{T}_{\nu\text{-SVM}}$ & $\hat{T}_{OCSVM}$ & $\hat{T}_{GP}$ & $\hat{T}_{KCM}$ & $\hat{T}_{ICM}$ \\
            
            \\
            \multicolumn{12}{c}{SIZE} \\
            \midrule
            $DGP_{1}$ (Bootstrap) & 0.058 & 0.057 & 0.002 & 0.000 & 0.000 && 0.051 & 0.056 & 0.004 & 0.000 & 0.000 \\
            \qquad \quad \,(Analytic) & [0.053] & [0.052] & - & - & - && [0.041] & [0.047] & - & - & - \\
            
            \\
            \multicolumn{12}{c}{POWER} \\
            \midrule
            $DGP_{2}$ (Bootstrap) & 0.230 & 0.255 & 0.029 & 0.000 & 0.000 && 0.455 & 0.443 & 0.117 & 0.001 & 0.000 \\
            \qquad \quad \,(Analytic) & [0.231] & [0.246] & - & - & - && [0.424] & [0.464] & - & - & - \\
            \\ 
            
            $DGP_{3}$ (Bootstrap) & 0.068 & 0.075 & 0.004 & 0.000 & 0.000 && 0.091 & 0.093 & 0.009 & 0.000 & 0.000 \\
            \qquad \quad \,(Analytic) & [0.052] & [0.065] & - & - & - && [0.084] & [0.076] & - & - & - \\
            \\ 
            
            $DGP_{4}$ (Bootstrap) & 0.519 & 0.517 & 0.111 & 0.000 & 0.000 && 0.831 & 0.825 & 0.439 & 0.011 & 0.000 \\
            \qquad \quad \,(Analytic) & [0.510] & [0.519] & - & - & - && [0.824] & [0.822] & - & - & - \\
            \\ 
            
            $DGP_{5}$ (Bootstrap) & 0.308 & 0.292 & 0.039 & 0.000 & 0.000 && 0.502 & 0.495 & 0.133 & 0.000 & 0.000 \\
            \qquad \quad \,(Analytic) & [0.252] & [0.225] & - & - & - && [0.524] & [0.498] & - & - & - \\
            \bottomrule
        \end{tabular}
    \label{tab:q20-5}
    \end{adjustbox}
\end{table}

Since our DGPs are sparse, we also investigate the finite-sample performance of the SVM-based test statistics when model estimators are obtained using LASSO. The results for $q = 10$ are presented in Tables \ref{tab:lasso-q10-nusvm} and \ref{tab:lasso-q10-tocsvm}, while the results for $q = 20$ are provided in Online Appendix C. Overall, we observe that the proposed SVM-based test statistics maintain accurate size control and high power under both analytic and multiplier bootstrapped critical values.
\newlength{\colwidth}
\setlength{\colwidth}{0.104\textwidth} 

\begin{table}[htbp]
    \centering
    \caption{Empirical sizes and powers of $\hat{T}_{\nu\text{-SVM}}$ estimated by LASSO with $q=10$}
    {\small 
    \begin{tabular}{
        @{} 
        >{\raggedright\arraybackslash}p{0.21\textwidth} 
        *{6}{>{\centering\arraybackslash}p{\colwidth}} 
        @{}
    }
        \toprule
        \multicolumn{1}{@{}>{\raggedright\arraybackslash}p{0.21\textwidth}}{$n$} 
        & \multicolumn{3}{c}{$n = 200$} 
        & \multicolumn{3}{c}{$n = 400$} \\
        \cmidrule(lr){2-4} \cmidrule(lr){5-7}
        & $10\%$ & $5\%$ & $1\%$ & $10\%$ & $5\%$ & $1\%$ \\
        \midrule
        \multicolumn{7}{c}{SIZE} \\
        \midrule
        \multicolumn{1}{@{}>{\raggedright\arraybackslash}p{0.21\textwidth}}{$DGP_{1}$ (Bootstrap)} & 0.095 & 0.041 & 0.009 & 0.094 & 0.043 & 0.012 \\
        \multicolumn{1}{@{}>{\raggedright\arraybackslash}p{0.21\textwidth}}{\qquad \quad \,(Analytic)}      & [0.082] & [0.037] & [0.009] & [0.103] & [0.050] & [0.011] \\
        \addlinespace
        \multicolumn{7}{c}{POWER} \\
        \midrule
        \multicolumn{1}{@{}>{\raggedright\arraybackslash}p{0.21\textwidth}}{$DGP_{2}$ (Bootstrap)} & 0.501 & 0.386 & 0.173 & 0.784 & 0.693 & 0.450 \\
        \multicolumn{1}{@{}>{\raggedright\arraybackslash}p{0.21\textwidth}}{\qquad \quad \,(Analytic)}      & [0.559] & [0.419] & [0.186] & [0.824] & [0.734] & [0.480] \\ \\
        \multicolumn{1}{@{}>{\raggedright\arraybackslash}p{0.21\textwidth}}{$DGP_{3}$ (Bootstrap)} & 0.585 & 0.440 & 0.213 & 0.831 & 0.735 & 0.476 \\
        \multicolumn{1}{@{}>{\raggedright\arraybackslash}p{0.21\textwidth}}{\qquad \quad \,(Analytic)}      & [0.589] & [0.454] & [0.209] & [0.819] & [0.724] & [0.474] \\ \\
        \multicolumn{1}{@{}>{\raggedright\arraybackslash}p{0.21\textwidth}}{$DGP_{4}$ (Bootstrap)} & 0.302 & 0.183 & 0.063 & 0.463 & 0.332 & 0.167 \\
        \multicolumn{1}{@{}>{\raggedright\arraybackslash}p{0.21\textwidth}}{\qquad \quad \,(Analytic)}      & [0.280] & [0.178] & [0.050] & [0.455] & [0.359] & [0.165] \\ \\
        \multicolumn{1}{@{}>{\raggedright\arraybackslash}p{0.21\textwidth}}{$DGP_{5}$ (Bootstrap)} & 0.400 & 0.290 & 0.111 & 0.690 & 0.576 & 0.327 \\
        \multicolumn{1}{@{}>{\raggedright\arraybackslash}p{0.21\textwidth}}{\qquad \quad \,(Analytic)}      & [0.421] & [0.296] & [0.115] & [0.660] & [0.532] & [0.296] \\
        \bottomrule
    \end{tabular}
    } 
    \label{tab:lasso-q10-nusvm}
\end{table}

\begin{table}[htbp]
    \centering
    \caption{Empirical sizes and powers of $\hat{T}_{OCSVM}$ estimated by LASSO with $q=10$}
    {\small 
    \begin{tabular}{
        @{} 
        >{\raggedright\arraybackslash}p{0.21\textwidth} 
        *{6}{>{\centering\arraybackslash}p{\colwidth}} 
        @{}
    }
        \toprule
        \multicolumn{1}{@{}>{\raggedright\arraybackslash}p{0.21\textwidth}}{$n$} 
        & \multicolumn{3}{c}{$n = 200$} 
        & \multicolumn{3}{c}{$n = 400$} \\
        \cmidrule(lr){2-4} \cmidrule(lr){5-7}
        & $10\%$ & $5\%$ & $1\%$ & $10\%$ & $5\%$ & $1\%$ \\
        \midrule
        \multicolumn{7}{c}{SIZE} \\
        \midrule
        \multicolumn{1}{@{}>{\raggedright\arraybackslash}p{0.21\textwidth}}{$DGP_{1}$ (Bootstrap)} & 0.090 & 0.047 & 0.011 & 0.091 & 0.040 & 0.010 \\
        \multicolumn{1}{@{}>{\raggedright\arraybackslash}p{0.21\textwidth}}{\qquad \quad \,(Analytic)}      & [0.089] & [0.038] & [0.007] & [0.099] & [0.048] & [0.011] \\
        \addlinespace
        \multicolumn{7}{c}{POWER} \\
        \midrule
        \multicolumn{1}{@{}>{\raggedright\arraybackslash}p{0.21\textwidth}}{$DGP_{2}$ (Bootstrap)} & 0.525 & 0.409 & 0.204 & 0.812 & 0.699 & 0.428 \\
        \multicolumn{1}{@{}>{\raggedright\arraybackslash}p{0.21\textwidth}}{\qquad \quad \,(Analytic)}      & [0.522] & [0.380] & [0.186] & [0.804] & [0.706] & [0.467] \\ \\
        \multicolumn{1}{@{}>{\raggedright\arraybackslash}p{0.21\textwidth}}{$DGP_{3}$ (Bootstrap)} & 0.576 & 0.449 & 0.209 & 0.818 & 0.727 & 0.498 \\
        \multicolumn{1}{@{}>{\raggedright\arraybackslash}p{0.21\textwidth}}{\qquad \quad \,(Analytic)}      & [0.572] & [0.447] & [0.191] & [0.823] & [0.733] & [0.510] \\ \\
        \multicolumn{1}{@{}>{\raggedright\arraybackslash}p{0.21\textwidth}}{$DGP_{4}$ (Bootstrap)} & 0.273 & 0.174 & 0.052 & 0.462 & 0.334 & 0.140 \\
        \multicolumn{1}{@{}>{\raggedright\arraybackslash}p{0.21\textwidth}}{\qquad \quad \,(Analytic)}      & [0.289] & [0.206] & [0.066] & [0.474] & [0.364] & [0.150] \\ \\
        \multicolumn{1}{@{}>{\raggedright\arraybackslash}p{0.21\textwidth}}{$DGP_{5}$ (Bootstrap)} & 0.413 & 0.305 & 0.117 & 0.657 & 0.544 & 0.278 \\
        \multicolumn{1}{@{}>{\raggedright\arraybackslash}p{0.21\textwidth}}{\qquad \quad \,(Analytic)}      & [0.401] & [0.292] & [0.122] & [0.661] & [0.531] & [0.282] \\
        \bottomrule
    \end{tabular}
    } 
    \label{tab:lasso-q10-tocsvm}
\end{table}

Reproducing kernel-based test statistics are known to be highly sensitive to the choice of kernel parameters, a phenomenon extensively documented in the nonparametric two-sample testing literature (see \cite{gretton2012optimal, sutherland2021generativemodelsmodelcriticism, liu2021learningdeepkernelsnonparametric}). To investigate the sensitivity of our proposed test statistics to the kernel parameter $\sigma$, we conduct a series of simulation studies using varying values of $\sigma$: $\sigma \in \{1, 2, 3, 4\}$. Figures \ref{sigma_varing_q10}--\ref{sigma_varing_q20} present the sizes and powers of the SVM-based, GP, and KCM tests under these different $\sigma$ values.

From these results, we draw two key observations. First, all specification test statistics exhibit sensitivity to the choice of kernel, consistent with findings in the existing literature on kernel-based testing. Second, SVM-based test statistics (with multiplier bootstrapped critical values) demonstrate robust performance across different $\sigma$ values. This robustness is particularly pronounced when the covariate dimension is high ($q = 20$).

\begin{figure}[htbp]
    \centering
    
    \begin{subfigure}[b]{0.45\textwidth}
        \centering
        \includegraphics[width=\textwidth]{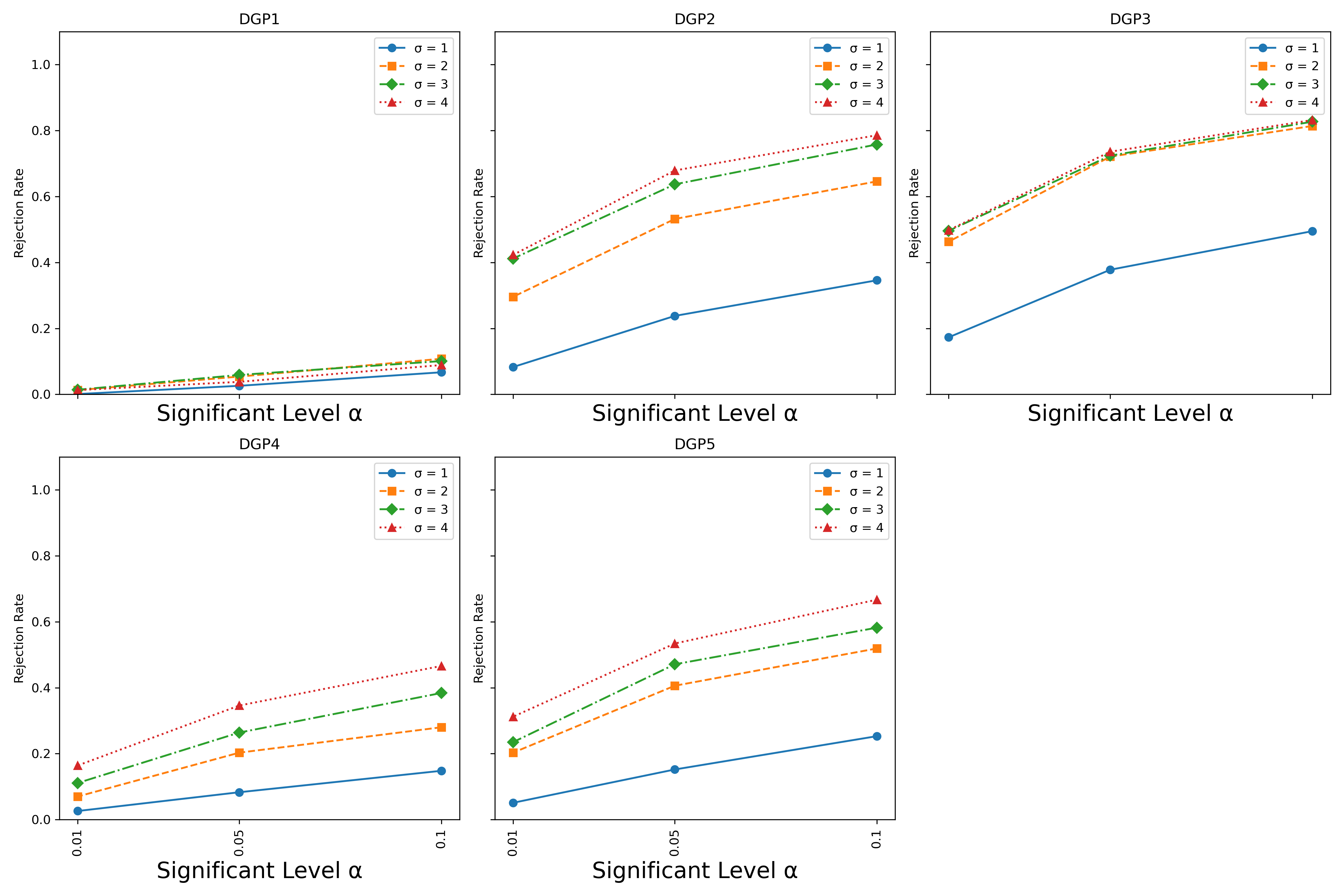} 
        \caption{$\hat{T}_{\nu-SVM}$ (bootstrap)}
        \label{fig:top-left}
    \end{subfigure}
    \hfill 
    \begin{subfigure}[b]{0.45\textwidth}
        \centering
        \includegraphics[width=\textwidth]{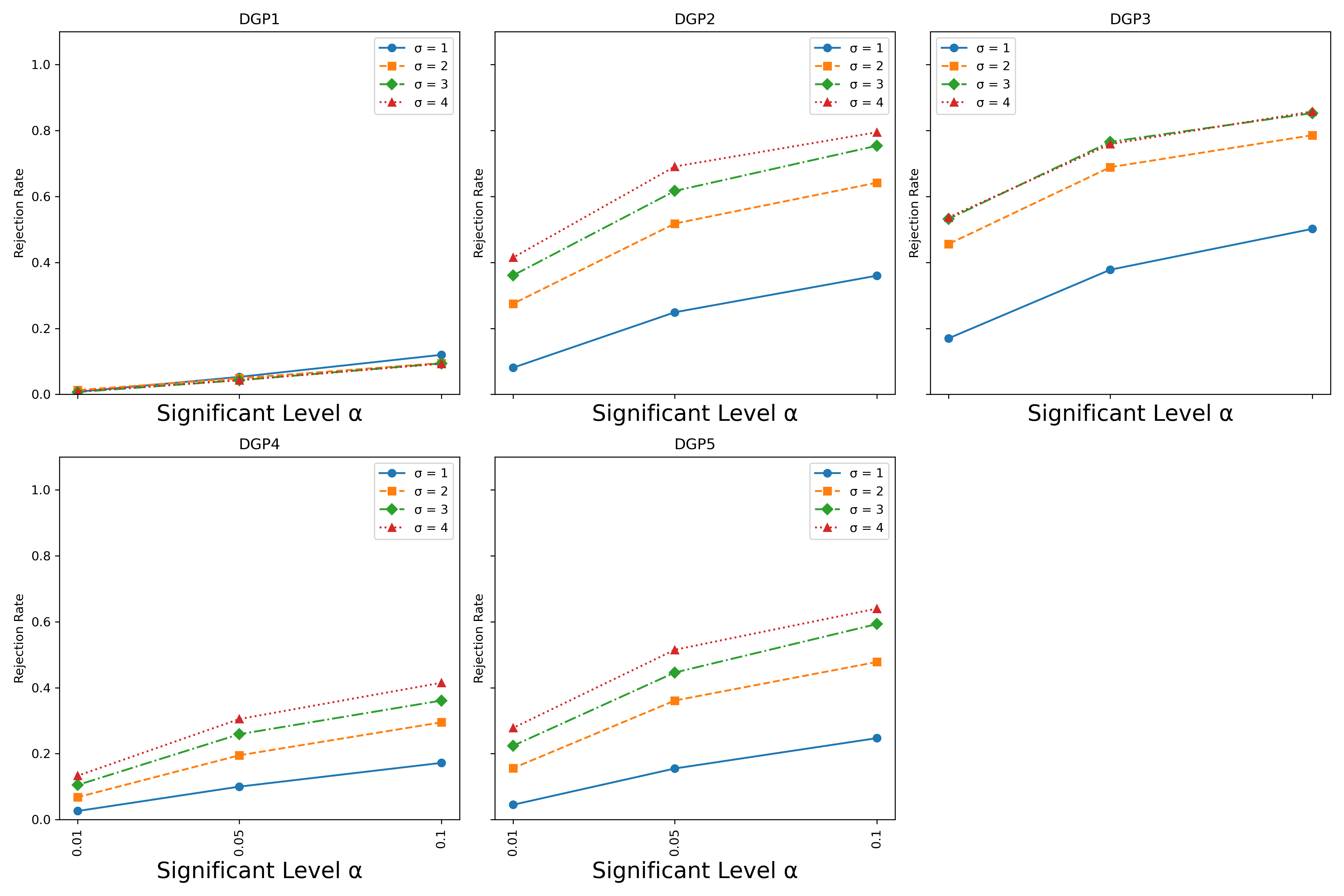} 
        \caption{$\hat{T}_{OCSVM}$ (bootstrap)}
        \label{fig:top-right}
    \end{subfigure}
    
    \vspace{1em} 
    
    \begin{subfigure}[b]{0.45\textwidth}
        \centering
        \includegraphics[width=\textwidth]{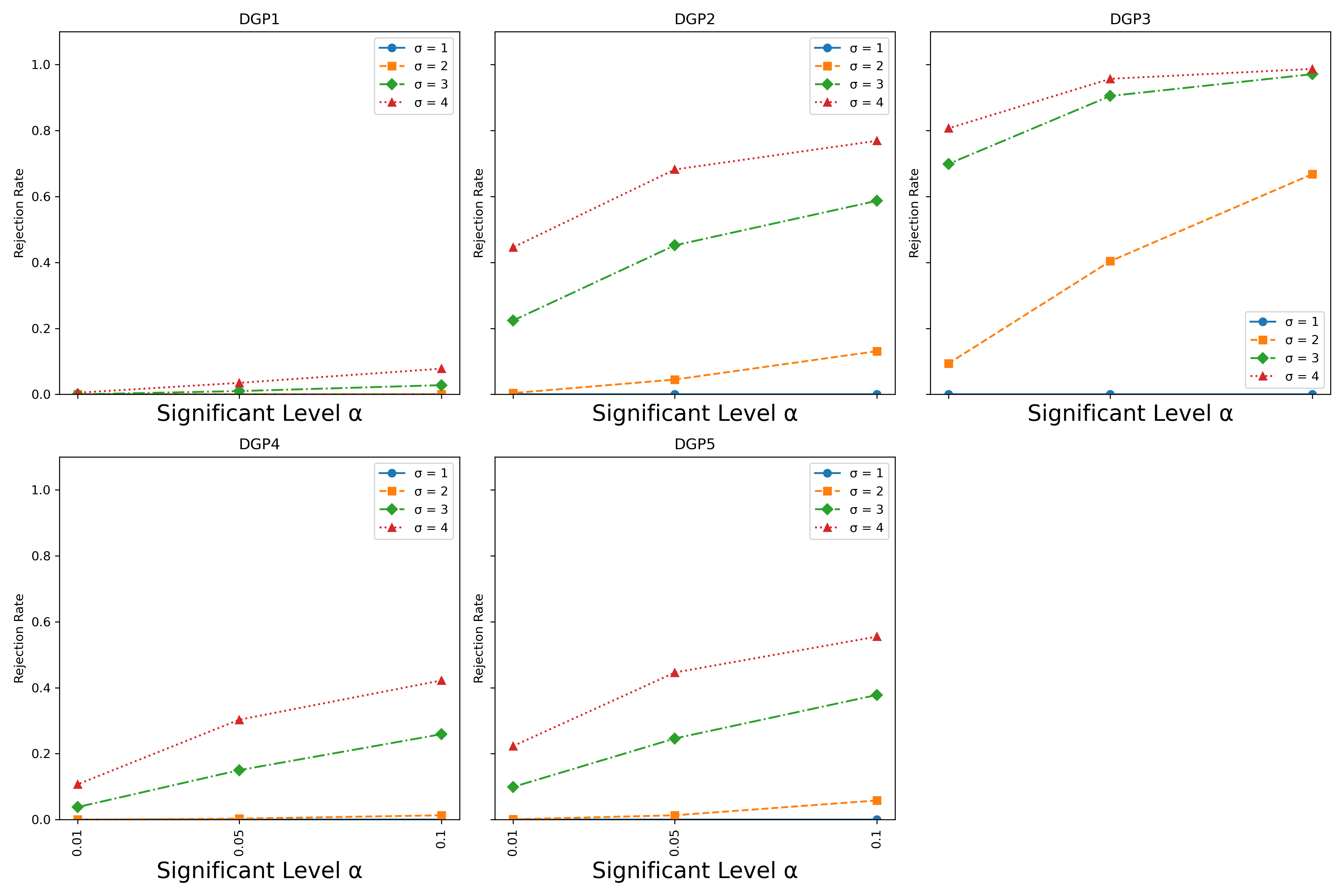} 
        \caption{$\hat T_{GP}$}
        \label{fig:bottom-left}
    \end{subfigure}
    \hfill %
    \begin{subfigure}[b]{0.45\textwidth}
        \centering
        \includegraphics[width=\textwidth]{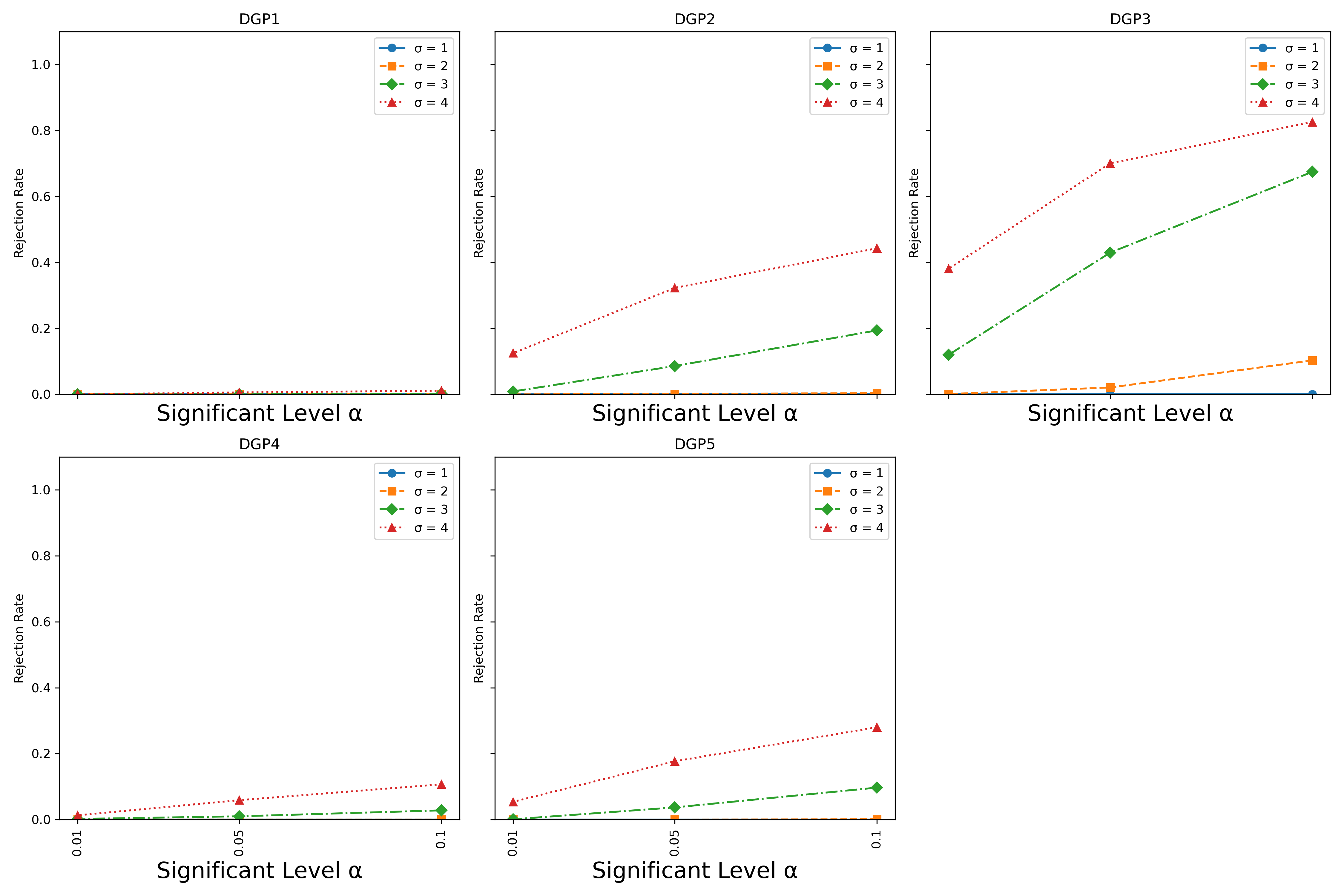} 
        \caption{$\hat T_{KCM}$}
        \label{fig:bottom-right}
    \end{subfigure}
    
    \caption{Size and Power of four test statistics with $q=10,N=400$ at different $\sigma$}
    \label{sigma_varing_q10}
\end{figure}

\begin{figure}[htbp]
    \centering
    
    \begin{subfigure}[b]{0.45\textwidth}
        \centering
        \includegraphics[width=\textwidth]{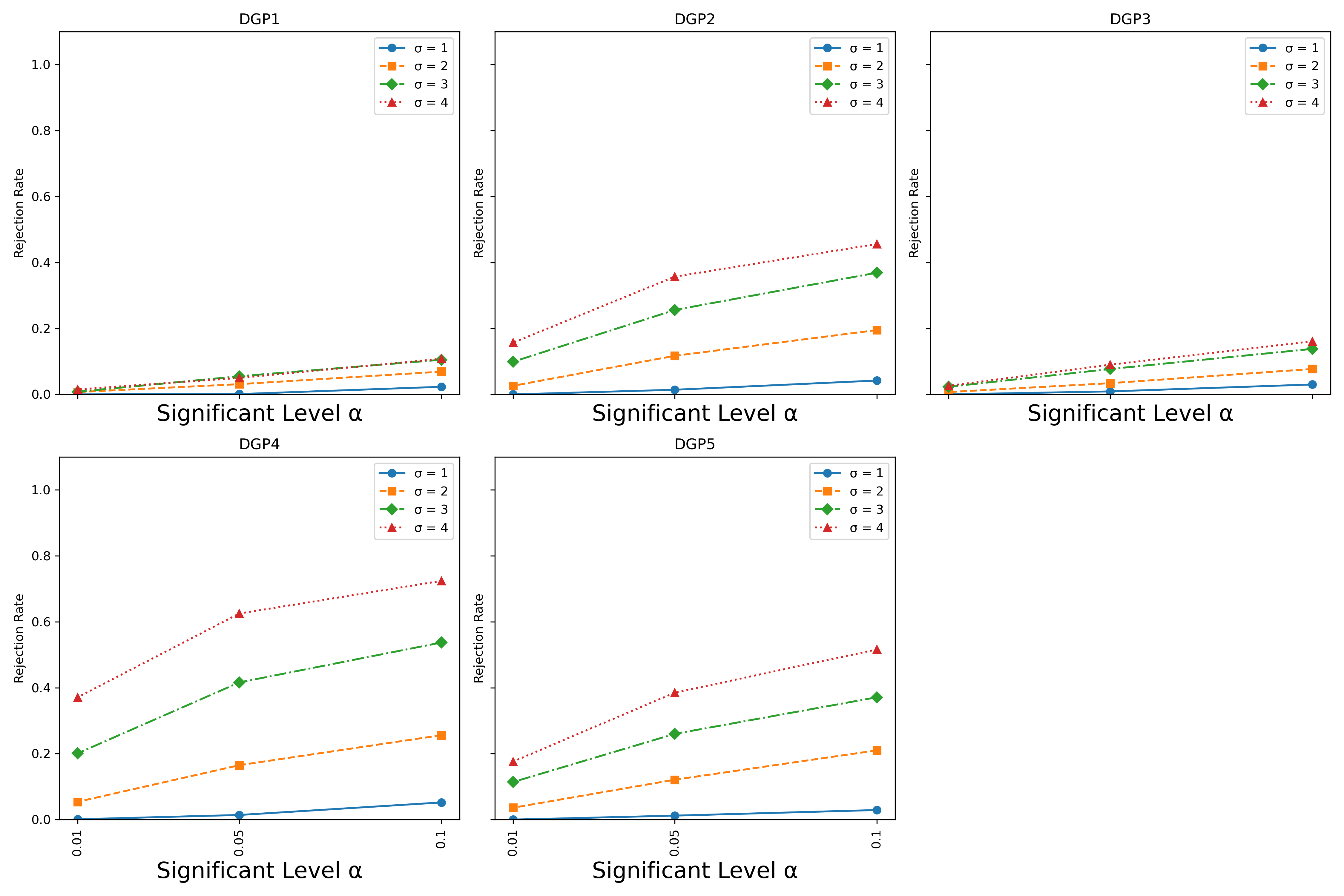} 
        \caption{$\hat{T}_{\nu-SVM}$ (bootstrap)}
        \label{fig:top-left}
    \end{subfigure}
    \hfill 
    \begin{subfigure}[b]{0.45\textwidth}
        \centering
        \includegraphics[width=\textwidth]{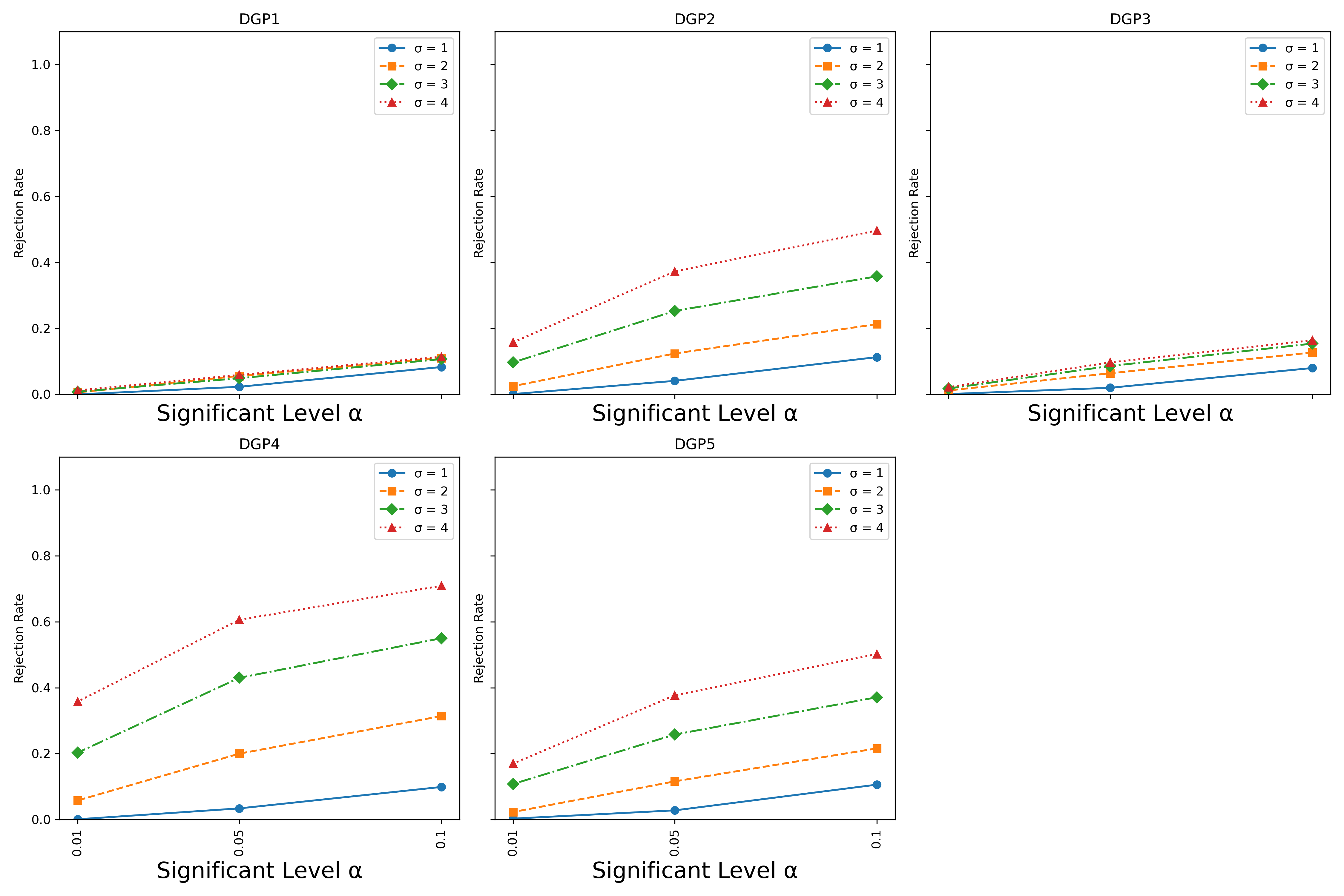} 
        \caption{$\hat{T}_{OCSVM}$ (bootstrap)}
        \label{fig:top-right}
    \end{subfigure}
    
    \vspace{1em} 
    
    \begin{subfigure}[b]{0.45\textwidth}
        \centering
        \includegraphics[width=\textwidth]{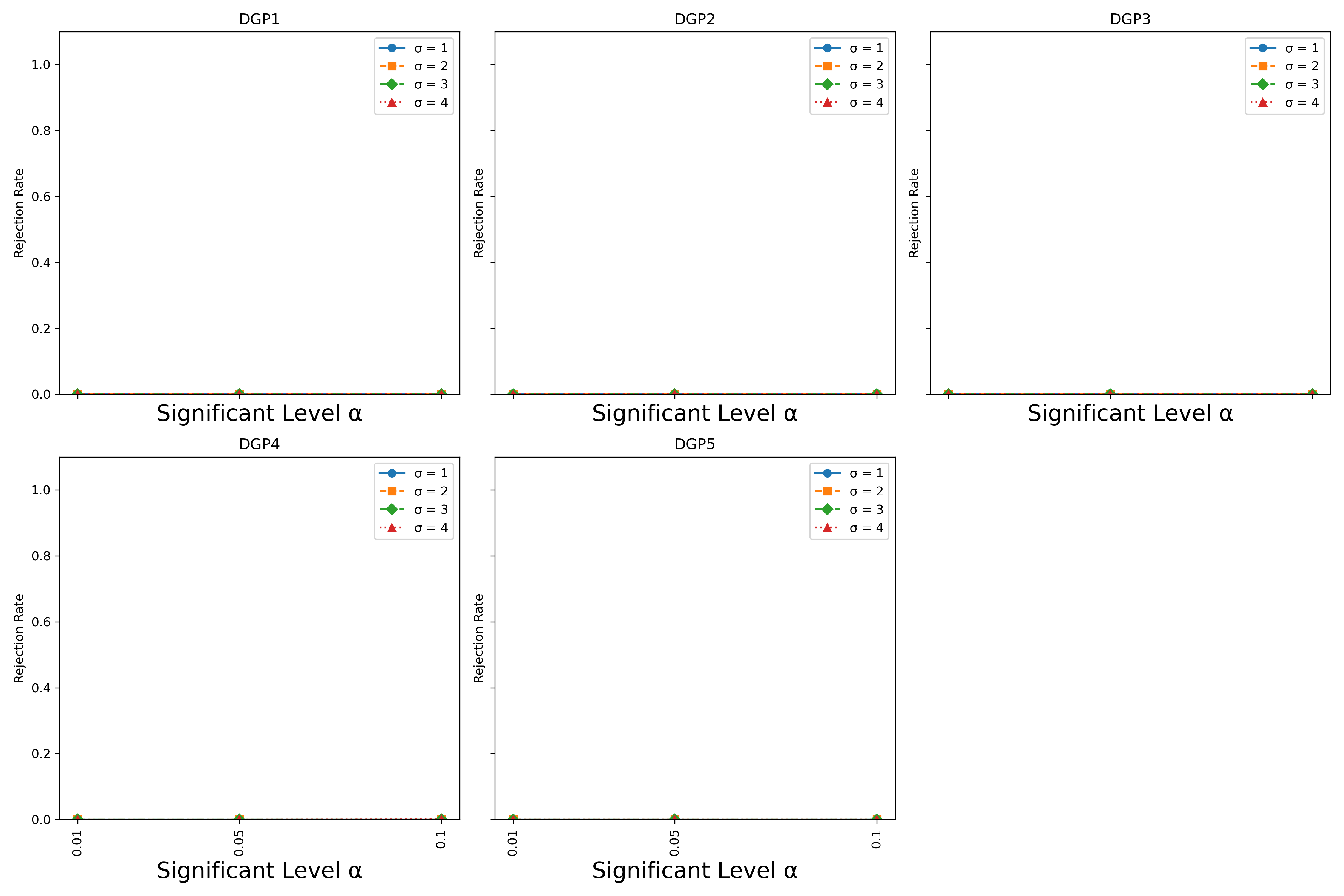} 
        \caption{$\hat T_{GP}$}
        \label{fig:bottom-left}
    \end{subfigure}
    \hfill %
    \begin{subfigure}[b]{0.45\textwidth}
        \centering
        \includegraphics[width=\textwidth]{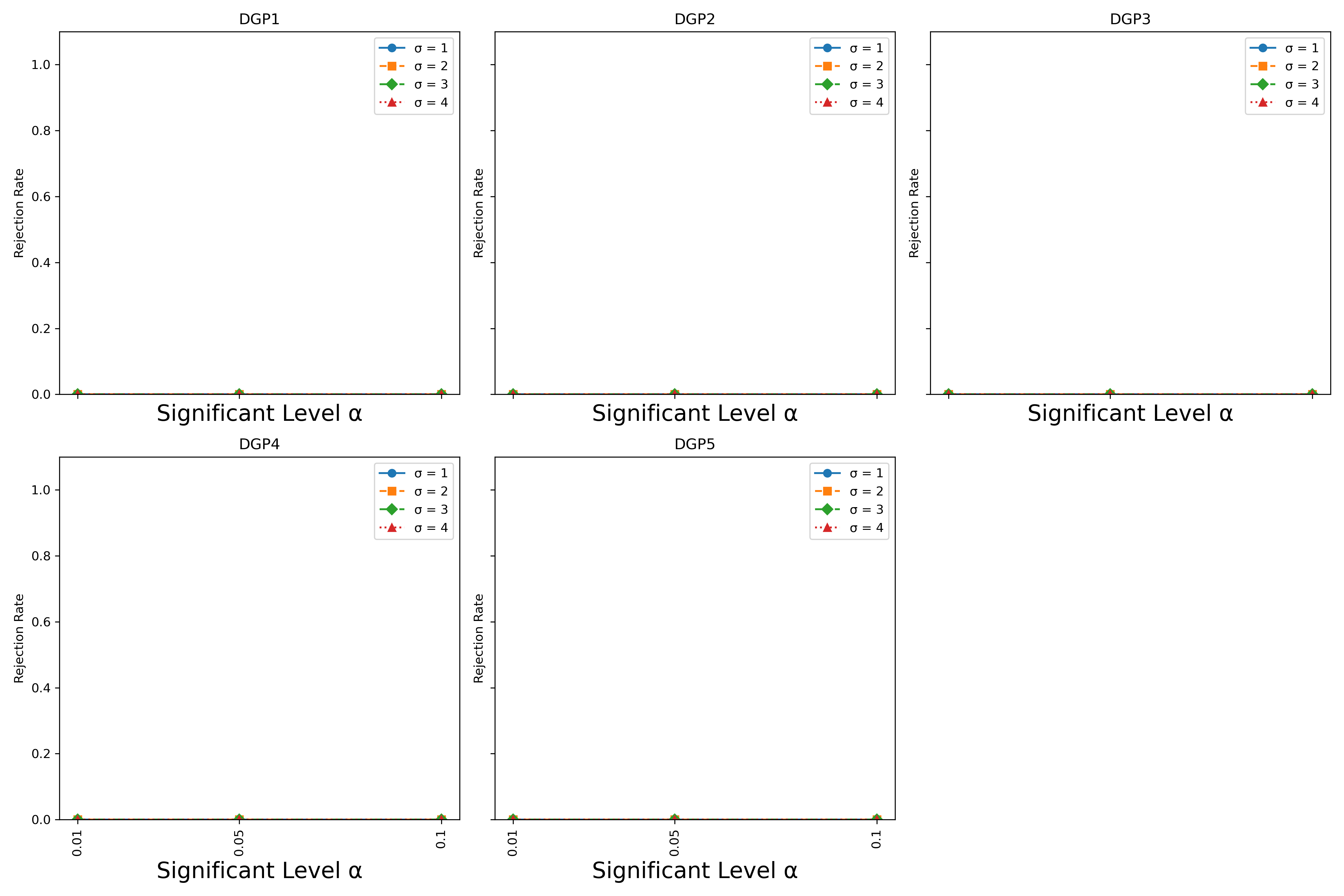} 
        \caption{$\hat T_{KCM}$}
        \label{fig:bottom-right}
    \end{subfigure}
    
    \caption{Size and Power of four test statistics with $q=20,N=400$ at different $\sigma$}
    \label{sigma_varing_q20}
\end{figure}

In addition to the finite sample performance, the computational complexity is also noteworthy. Here we document the run times in \(\mathrm{DGP}_2\) as an illustration. We focus on the case where the sample size is \(n = 400\) and the covariate dimension is \(q = 20\). The records are collected by running tests in Python 3.13.2 on a modern desktop computer (AMD Ryzen 9 9950X CPU, 128GB RAM), with parallel computing enabled (using the \texttt{joblib.Parallel} routine with \(32\) CPU cores). In Table \ref{tab:running-time}, we report the time costs for the \(1000\) repeated experiments for each test with \(500\) bootstrap samples when required. Clearly, the proposed SVM-based tests are more computationally efficient than the other bootstrap-based tests. The analytic implementation of the SVM-based tests is even faster, as it does not require bootstrapping.

\begin{table}[htbp]
    \centering
    \caption{Running Time (seconds) Comparison for 1000 Test Repetitions}
    {\footnotesize 
    \begin{tabular}{
        @{} 
        >{\raggedright\arraybackslash}p{0.18\textwidth} 
        *{5}{>{\centering\arraybackslash}p{0.12\textwidth}} 
        @{}
    }
        \toprule
        & $\hat{T}_{\nu\text{-SVM}}$ & $\hat{T}_{OCSVM}$ & $\hat{T}_{GP}$ & $\hat{T}_{KCM}$ & $\hat{T}_{ICM}$ \\
        \midrule
        Bootstrap & 1.18 & 1.02 & 52.80 & 47.61 & 56.25 \\
        Analytic  & 0.18 & 0.16 & -     & -     & -     \\
        \bottomrule
    \end{tabular}
    } 
    \label{tab:running-time}
\end{table}

\subsection{Empirical Studies}
We apply our proposed SVM-based specification tests to two datasets from the UCI Machine Learning Repository \citep{lichman2017uci}. The first dataset, ``Wine Quality'', comprises 11 covariates and 4,898 observations. The response variable is wine quality, scored on a scale of 3 to 9. We use linear regression to model wine quality based on the covariates. 

The second dataset, ``Students' Dropout'', includes 36 covariates and 4,424 observations. The binary response variable indicates whether a student has dropped out of school. We employ the Probit model to investigate the probability of dropout given the covariates.

Given the relatively large sample sizes, we randomly split each dataset into 60\% training data and 40\% testing data. Table \ref{tab:realdata} presents $p$-values from different testing procedures. 
\begin{table}[htbp]
    \centering
    \caption{Bootstrap $p$-values of different test statistics for real data}
    \label{tab:realdata}
    \begin{tabularx}{\textwidth}{@{}l>{\centering\arraybackslash}X>{\centering\arraybackslash}X@{}}
        \toprule
        Dataset & Wine Quality & Students' Dropout \\
        \midrule
        $\hat T_{\nu-SVM}$ (Bootstrap) & 0.000 & 0.012 \\
        \qquad \quad \quad (Analytic) & 0.063 & 0.238 \\
        $\hat T_{OCSVM}$ (Bootstrap) & 0.000 & 0.002 \\
        \qquad \quad \quad (Analytic) & 0.124 & 0.018 \\
        $\hat T_{GP}$ & 0.000 & 0.000 \\
        $\hat T_{KCM}$ & 0.000 & 0.002 \\
        $\hat T_{ICM}$ & 0.000 & 0.378 \\
        \bottomrule
    \end{tabularx}
\end{table}

For the Wine Quality dataset, all bootstrap-based tests ($\hat T_{\nu-SVM}$ and $\hat{T}_{OCSVM}$ with bootstrap, $\hat{T}_{GP}$, $\hat{T}_{KCM}$, and $\hat{T}_{ICM}$) unanimously reject the null hypothesis (p-value = 0.000), indicating that the linear regression model is misspecified for this dataset. Interestingly, the SVM-based tests with analytic critical values fail to reject the null hypothesis at the 5\% significance level, showing a notable difference from its bootstrap counterpart.

For the Students' Dropout dataset, Most tests ($\hat T_{\nu-SVM}$ and $\hat{T}_{OCSVM}$ with bootstrap, $\hat{T}_{GP}$, and $\hat{T}_{KCM}$) reject the null hypothesis at the 5\% significance level, suggesting that the logistic regression model is inadequate for this dataset. The ICM ($\hat{T}_{ICM}$) and $\nu-$SVM ($\hat T_{\nu-SVM}$) tests are the only two that fail to reject the null hypothesis, contradicting the other test results.

The discrepancy between bootstrap and analytic implementations of the SVM-based tests confirms the simulation study findings that bootstrap methods generally provide better finite-sample performance. The ICM test's inconsistency with other tests for the Students' Dropout dataset might be due to its lower power in detecting certain types of misspecification, especially in higher-dimensional settings (this dataset has 36 covariates).

Overall, these results suggest that both datasets likely require more complex models than simple linear or logistic regression to accurately capture the underlying relationships between predictors and response variables.

\section{Conclusion}\label{sec:conclusion}

This paper proposes a novel framework for enhancing the power of specification tests in parametric models by strategically integrating Support Vector Machines (SVMs) for direction learning. By addressing limitations in traditional methods, such as ICM and KCM tests, our approach focuses on identifying good projection directions in Reproducing Kernel Hilbert Spaces (RKHS) to enhance the test's ability to detect deviations from the null hypothesis. The key innovation lies in two SVM-based mechanisms: (1) maximizing the discrepancy between nonparametric and parametric classes via a margin-maximizing hyperplane, and (2) maximizing the separation between residuals and the origin using one-class SVM. These mechanisms target the signal-to-noise ratio in the test statistic, ensuring superior power against arbitrary alternatives.  

Theoretical analysis establishes consistency under the alternative hypothesis, while simulations demonstrate significant power gains over existing methods, particularly in high-dimensional settings. The use of multiplier bootstrap for critical value computation ensures reliable inference even with a small sample size. Empirical studies on real datasets further validate the proposed methodology, highlighting its effectiveness in detecting model misspecification.

By leveraging SVMs to learn optimal projection directions, this work advances the toolkit for model validation, offering a computationally efficient, omnibus solution with a pivotal chi-square asymptotic distribution. Future research could explore adaptive kernel selection methods (e.g., the AdaBoost algorithm) to further enhance test power. Another interesting direction would be to construct test statistics based on the classification accuracy of the learned SVM, an approach gaining traction in the nonparametric two-sample testing literature \citep{kim2021classification,hediger2022use,lopez2016revisiting}. These extensions could further capitalize on the connection between classification methodology and specification testing established in this work.

\bibliography{paper-ref}

\newpage
\appendix
\begin{center}
    \Large\textbf{Online Appendix}
\end{center}

\section{Proofs}

\subsection{Proof of Theorem 1}
The ``if'' part of Theorem 1 is straightforward. We only need to show the ``only if'' part. 

By Mercer's theorem, we have
\[
    k(x_i,x_j) = \sum_{j=1}^\infty \lambda_j \phi_j(x_i) \phi_j(x_j),
\]
where $\lambda_j$ and $\phi_j$ are the eigenvalues and eigenfunctions of the integral operator $T$ defined by $Tf(x) = \int_{\Omega} k(x,y) f(y) d P(y)$, where $P$ is a measure of the domain of $y$, and $\Omega$ is the support. 

Thus,
\[
	\begin{split}
		\mu_{\theta_0}(\cdot)^\dagger &= \mathbb{E}\left(\varepsilon_{\theta_0}^\dagger \sum_{j=1}^\infty \lambda_j \phi_j(X^\dagger)\phi_j(\cdot)\right) \\
		& = \sum_{j=1}^\infty \lambda_j \mathbb{E}(\varepsilon_{\theta_0}^\dagger \phi_j(X^\dagger)) \phi_j(\cdot)
	\end{split}
\]

Note that Mercer's theorem also states that $\phi_l(\cdot)$ forms an orthonormal basis of $L^2(P)$, and the eigenvalues $\lambda_l$ are non-negative and decreasing.

Since $\mathbb{E}(\varepsilon_{\theta_0}^\dagger \mid X^\dagger) \in L^2(P)$, 
\[
    \mathbb{E}(\varepsilon_{\theta_0}^\dagger \mid X^\dagger) = \sum_{j=1}^\infty \gamma_j \phi_j(X^\dagger),
\]
where $\{\gamma_j\}_{j=1}^\infty$ are the coefficients of the expansion. 

For a learned index set $\mathbb{S}$, we have:
\[
	\begin{split}
		\mu_{\theta_0,\mathbb{S},k}^\dagger = \left \langle \sum_{j \in \mathbb{S}} \eta_j k(x_j,\cdot), \mu_{\theta_0}^\dagger \right \rangle_{\mathcal{H}_k} & = \sum_{i=1}^\infty \sum_{j \in \mathbb{S}} \eta_j \left \langle k(x_j,\cdot), \phi_i \right \rangle_{\mathcal{H}_k} \lambda_i \mathbb{E}(\varepsilon_{\theta_0}^\dagger \phi_i(X^\dagger)) \\
		& = \sum_{j \in \mathbb{S}}  \sum_{i=1}^\infty \eta_j \lambda_i \phi_i(x_j) \mathbb{E}(\varepsilon_{\theta_0}^\dagger \phi_i(X^\dagger)) \\
		& = 0
	\end{split} 
\]
The above equation holds for any $J$ and non-zero weight $\eta_j$, and due to the properties of orthonormal bases $\{\phi_j(\cdot)\}_{j=1}^\infty$, we conclude $\varepsilon_{\theta_0}^\dagger$ is orthogonal to all the basis functions almost surely:
\[
    \mathbb{E}(\varepsilon_{\theta_0}^\dagger \phi_j(X^\dagger))= 0, \,\,\,\forall j.
\]

To find the coefficients $\{\gamma_j\}_{j=1}^\infty$, we use the orthogonality conditions as given below:
\[
    \begin{split}
        \mathbb{E}(\varepsilon_{\theta_0}^\dagger \phi_j(X^\dagger)) & = \mathbb{E}\left[ \mathbb{E}(\varepsilon_{\theta_0}^\dagger \mid X^\dagger) \phi_j(X^\dagger) \right] \\
        & = \mathbb{E}\left[ \sum_{l=1}^\infty \gamma_l \phi_l(X^\dagger) \phi_j(X^\dagger) \right] \\
        & = \gamma_j  \mathbb{E}\left[ \phi_j(X^\dagger)^2 \right] \\
        & = \gamma_j = 0, \quad \forall j.
    \end{split}
\]
Hence, $\mathbb{E}(\varepsilon_{\theta_0} \mid X) = 0$.

\subsection{Proof of Theorem 4}
Let the residuals under the local alternative be denoted as $\{\tilde{\varepsilon}_i\}_{i=1}^{n}$, where for each $i$,
\[
	\tilde{\varepsilon}_i = \varepsilon_{\theta_0,i}^\dagger + \frac{R(x_i^\dagger)}{\sqrt{n}}.
\]
The corresponding mean difference element at finite points becomes: 
\[
	\begin{split}
		\tilde{\mu}_{\theta_0,\mathbb{S},k}^\dagger &= \mu_{\theta_0,\mathbb{S},k}^\dagger + \sum_{j \in \mathbb{S}} \eta_j \frac{\mathbb{E}\left[R(X^\dagger)k(X^\dagger,x_j)\right]}{\sqrt{n}}\\
		& = \mu_{\theta_0,\mathbb{S},k}^\dagger + \sum_{j \in \mathbb{S}} \eta_j \frac{\Delta_j}{\sqrt{n}},
	\end{split}
\]
which can be estimated by 
\[
	\hat{\mu}_{\theta_0,\mathbb{S},k}^\dagger + \sum_{j \in \mathbb{S}} \eta_j\frac{1}{n} \sum_{i=1}^{n} \frac{R(x_i^\dagger) k(x_i^\dagger,x_j)	}{\sqrt{n}}.
\]
After standardization and rescaling by the convergence speed, we have: 
\[
	\sqrt{n} \frac{\hat{\mu}_{\theta_0,\mathbb{S},k}^\dagger}{\hat{\sigma}_{\theta_0,\mathbb{S},k}^\dagger} + \sum_{j \in \mathbb{S}} \eta_j \frac{1}{n} \sum_{i=1}^{n} \frac{R(x_i^\dagger) k(x_i^\dagger,x_j)}{\hat{\sigma}_{\theta_0,\mathbb{S},k}^\dagger} \xrightarrow{d} \mathcal{N}\left(\frac{\sum_{j\in \mathbb{S}}\eta_j \Delta_j}{\sigma_{\theta_0,\mathbb{S},k}^\dagger},1\right),
\]
by the central limit theorem, the law of large numbers, and Slutsky's theorem.

\subsection{Proof of Theorem 5}
The first equality is trivial,  to show the second equality, note that by Assumptions 1 and 2, and the mean value theorem, we have
\[
	\boldsymbol{\hat \varepsilon}^\dagger = \boldsymbol{\varepsilon}_{\theta_0}^\dagger + \mathbf{g(\bar \theta)} (\hat \theta -\theta_0) = \boldsymbol{\varepsilon}_{\theta_0}^\dagger + \mathbf{\hat g} (\hat \theta -\theta_0) +O_p(n^{-2\alpha})
\] 
where $\hat{\mathbf{g}}(\bar \theta)$ is a $n \times d$ matrix of scores whose $i$th row is given by $(\hat g_{i}^\dagger)^\top = (\nabla_{\theta} \varepsilon_{\theta}^\dagger|_{\theta = \bar \theta})^\top$, and $\bar \theta$ is a value between $\theta_0$ and $\hat \theta$. Thus, 
\[
	\begin{split}
		\boldsymbol{\hat \varepsilon}_p^\dagger & = \hat{\boldsymbol{\Pi}}^\dagger \boldsymbol{\hat \varepsilon}\\        
        & =\hat{\boldsymbol{\Pi}}^\dagger ( \boldsymbol{\varepsilon}_{\theta_0}^\dagger + \mathbf{\hat g} (\hat \theta -\theta_0) +O_p(n^{-2\alpha})) \\
		& = \hat{\boldsymbol{\Pi}}^\dagger  \boldsymbol{\varepsilon}_{\theta_0}^\dagger+ \hat{\boldsymbol{\Pi}}^\dagger O_p(n^{-2\alpha})\\
        & = \boldsymbol{\varepsilon}_{p,\theta_0}^\dagger +\hat{\boldsymbol{\Pi}}^\dagger O_p(n^{-2\alpha})
	\end{split}
\]
Putting everything together, we have 
\[
  \begin{split}
	\frac{1}{n} (\boldsymbol{\hat \varepsilon}_p^\dagger)^\top \mathbf{K}(\mathbf{X}^\dagger,\cdot)& = \frac{1}{n} (\boldsymbol{\varepsilon}_{p,\theta_0}^\dagger)^\top \mathbf{K}(\mathbf{X}^\dagger,\cdot) + \frac{1}{n} (\hat{\boldsymbol{\Pi}}^\dagger O_p(n^{-2\alpha}))^\top \mathbf{K}(\mathbf{X}^\dagger,\cdot) \\
	& = \frac{1}{n} (\boldsymbol{\varepsilon}_{p,\theta_0}^\dagger)^\top \mathbf{K}(\mathbf{X}^\dagger,\cdot) + O_p(n^{-2\alpha}).
  \end{split}
\]
\subsection{Proof of Theorem 6}
\[
	\begin{split}
		\sqrt{n}\hat{\mu}_{\hat{\theta}, \mathbb{S},k_p}^* &= \frac{1}{\sqrt{n}} \sum_{j \in \mathbb{S}} \eta_j (\boldsymbol{\hat \varepsilon}_{p}^*)^\top \mathbf{K}(\mathbf{X}^\dagger,x_j) \\
		& = \frac{1}{\sqrt{n}} \sum_{j \in \mathbb{S}} \eta_j \left((\boldsymbol{\hat \varepsilon}^\dagger \odot \mathbf{V}) - \mathbf{\hat g} \left( \mathbf{\hat g}^\top \mathbf{\hat g} \right)^{-1} \mathbf{\hat g}^\top (\boldsymbol{\hat \varepsilon}^\dagger \odot \mathbf{V}) \right)^\top \mathbf{K}(\mathbf{X}^\dagger,x_j) \\
		& = \frac{1}{\sqrt{n}} \sum_{j \in \mathbb{S}} \eta_j (\boldsymbol{\hat \varepsilon}^\dagger \odot \mathbf{V})^\top \mathbf{K}(\mathbf{X}^\dagger,x_j) - \frac{1}{\sqrt{n}} \sum_{j \in \mathbb{S}} \eta_j \mathbf{\hat g} \left( \mathbf{\hat g}^\top \mathbf{\hat g} \right)^{-1} \mathbf{\hat g}^\top (\boldsymbol{\hat \varepsilon}^\dagger \odot \mathbf{V})^\top \mathbf{K}(\mathbf{X}^\dagger,x_j) \\
		& = S_{1n}^* + S_{2n}^*.
	\end{split}
\]
Note in both $S_{1n}^*$ and $S_{2n}^*$ for every $j$, it follows from the consistency of $\hat \theta$ to $\theta_0$
\[
	\frac{1}{\sqrt{n}}(\boldsymbol{\hat \varepsilon}^\dagger \odot \mathbf{V})^\top \mathbf{K}(\mathbf{X}^\dagger,x_j)  = \frac{1}{\sqrt{n}}(\boldsymbol{\varepsilon}_{\theta_0}^\dagger \odot \mathbf{V})^\top \mathbf{K}(\mathbf{X}^\dagger,x_j) +o_p(1),
\]
and 
\[
	\frac{1}{\sqrt{n}} \mathbf{\hat g} \left( \mathbf{\hat g}^\top \mathbf{\hat g} \right)^{-1} \mathbf{\hat g}^\top (\boldsymbol{\hat \varepsilon}^\dagger \odot \mathbf{V})^\top \mathbf{K}(\mathbf{X}^\dagger,x_j)  =\frac{1}{\sqrt{n}} \mathbf{\hat g} \left( \mathbf{\hat g}^\top \mathbf{\hat g} \right)^{-1} \mathbf{\hat g}^\top  (\boldsymbol{\varepsilon}_{\theta_0}^\dagger \odot \mathbf{V})^\top \mathbf{K}(\mathbf{X}^\dagger,x_j)+o_p(1).
\]
Thus, 
\[
	\begin{split}
		\sqrt{n}\hat{\mu}_{\hat{\theta}, \mathbb{S},k_p}^* & = \sum_{j \in \mathbb{S}} \eta_j \frac{1}{\sqrt{n}}  \sum_{i=1}^n \varepsilon_{\theta_0,i}^\dagger v_i \boldsymbol{\Pi} k(x_i^\dagger,x_j)  +o_p(1) \\
        & = \sum_{j \in \mathbb{S}} \eta_j \frac{1}{\sqrt{n}}  \sum_{i=1}^n \varepsilon_{\theta_0,i}^\dagger v_i  k_p(x_i^\dagger,x_j)  +o_p(1)\\
		& = \sqrt{n}\hat{\mu}_{\theta_0, \mathbb{S},k_p}^*+o_p(1).
	\end{split}
\]
The rest of the proof then follows from the multiplier central limit theorem; see \cite{van1996weak}.

\section{Test Statistics Construction and Bootstrap Procedure}
Both the two-class SVM and one-class SVM-based test statistics involve the shift transformation. In practice, we find the following two transformations to be effective:
\[
    \begin{split}
        & e = \max\{|z_{p,i}|\}_{i=1}^{2n} + 0.1, \quad \text{for $\nu-$SVM}, \\
        & e = \max\{|\hat{\varepsilon}_{p,i}|\}_{i=1}^n + 0.1, \quad \text{for one-class SVM}.
    \end{split}
\]

Algorithm 1 and Algorithm 2 present comprehensive summaries of the proposed testing procedures, and Algorithm 3 outlines the bootstrap procedure for obtaining critical values.

\begin{algorithm}[H]
    \caption{SVM-based Test Statistic}
    \begin{algorithmic}[1]
        \REQUIRE Data $(X,Y)$, kernel $k$, training sample size $m$, testing sample size $n$

        \textbf{Step 1: Sample Splitting}
        \STATE Randomly split data into training set $(X, Y)$ and testing set $(X^{\dagger}, Y^{\dagger})$. 

        \textbf{Step 2: Projection on Training Set}
        \STATE Estimate $\hat{\theta}$ using $(X, Y)$.
        \STATE Compute the Gram matrix $\mathbf{K}$ with entries $K_{i,j} = k(x_i,x_j)$.

        \textbf{Step 3: Train OCSVM Model}
        \STATE Obtain $Y_p$ and $\mathcal{M}_{\hat{\theta},p}(X)$: 
        \[
        Y_p = Y - \mathbf{\hat{g}}((\mathbf{\hat{g}})^\top \mathbf{\hat{g}})^{-1}(\mathbf{\hat{g}})^\top Y,
        \]
        \[
        \mathcal{M}_{\hat{\theta},p}(X) = \mathcal{M}_{\hat{\theta}}(X) - \mathbf{\hat{g}}((\mathbf{\hat{g}})^\top \mathbf{\hat{g}})^{-1}(\mathbf{\hat{g}})^\top \mathcal{M}_{\hat{\theta}}(X).
        \]
        \STATE Shift the data points: 
        \[
        \tilde{Y}_p = Y_p + e > 0, \quad \widetilde{\mathcal{M}}_{\hat{\theta},p}(X) = \mathcal{M}_{\hat{\theta},p}(X) + e > 0.
        \]
        \STATE Train SVM on $\tilde{Y}_p, \widetilde{\mathcal{M}}_{\hat{\theta},p}(X)$ with kernel $\mathbf{K}$ to obtain $w^*$.

        \textbf{Step 4: Construct Test Statistics}
        \STATE Compute residuals $\boldsymbol{\hat{\varepsilon}}^{\dagger} = Y^{\dagger} - \mathcal{M}_{\hat{\theta}}(X^{\dagger})$.
        \STATE Compute projected residuals:
        \[
        \boldsymbol{\hat{\varepsilon}}_p^{\dagger} = \boldsymbol{\hat{\varepsilon}}^{\dagger} - \mathbf{\hat{g}}^{\dagger}((\mathbf{\hat{g}}^{\dagger})^\top \mathbf{\hat{g}}^{\dagger})^{-1}(\mathbf{\hat{g}}^{\dagger})^\top \boldsymbol{\hat{\varepsilon}}^{\dagger}.
        \]
        \STATE Compute t-statistic: $\hat{T}_{\hat{\theta},\mathbb{S}^*,k_p}$.
        
    \end{algorithmic}
\end{algorithm}

\begin{algorithm}[H]
    \caption{OCSVM-based Test Statistic}
    \begin{algorithmic}[1]
        \REQUIRE Data $(X,Y)$, kernel $k$, training sample size $m$, testing sample size $n$

        \textbf{Step 1: Sample Splitting}
        \STATE Randomly split data into training set $(X, Y)$ and testing set $(X^{\dagger}, Y^{\dagger})$. 

        \textbf{Step 2: Projection on Training Set}
        \STATE Estimate $\hat{\theta}$ using $(X, Y)$.
        \STATE Compute residuals $\boldsymbol{\hat{\varepsilon}} = Y - \mathcal{M}_{\hat{\theta}}(X)$ and shift to positive:
        \[
        \boldsymbol{\tilde{\varepsilon}} = \boldsymbol{\hat{\varepsilon}} + \boldsymbol{c} > \boldsymbol{0}.
        \]
        \STATE Compute the Gram matrix $\mathbf{K}$ with entries $K_{i,j} = k(x_i,x_j)$.

        \textbf{Step 3: Train OCSVM Model}
        \STATE Train OCSVM on $\boldsymbol{\tilde{\varepsilon}}$ with kernel $\mathbf{K}$ to obtain dual coefficients $w^*$.

        \textbf{Step 4: Construct Test Statistics}
        \STATE Compute residuals $\boldsymbol{\hat{\varepsilon}}^{\dagger} = Y^{\dagger} - \mathcal{M}_{\hat{\theta}}(X^{\dagger})$.
        \STATE Compute projected residuals:
        \[
        \boldsymbol{\hat{\varepsilon}}_p^{\dagger} = \boldsymbol{\hat{\varepsilon}}^{\dagger} - \mathbf{\hat{g}}^{\dagger}((\mathbf{\hat{g}}^{\dagger})^\top \mathbf{\hat{g}}^{\dagger})^{-1}(\mathbf{\hat{g}}^{\dagger})^\top \boldsymbol{\hat{\varepsilon}}^{\dagger}.
        \]
        \STATE Compute t-statistic: $\hat{T}_{\hat{\theta},\mathbb{S}^*,k_p}$.
    
    \end{algorithmic}
\end{algorithm}

\begin{algorithm}[H]
    \caption{Multiplier Bootstrap Procedure}
    \label{alg:multiplier_bootstrap_binom}
    \begin{algorithmic}[1]
        \STATE \textbf{Input:}
        \STATE Kernel Matrix on test data $\mathbf{K}^{\dagger}$, a $n \times m$ matrix with element $(\mathbf{K}^\dagger)_{i,j} = k(x_i^\dagger,x_j)$, number of bootstrap samples $B$, significance level $\alpha$.
        \STATE \textbf{Output:} Bootstrap critical values $C_{\alpha,B}$.

        \STATE Initialize an array $\text{stat\_kerb}$ of size $B$ to store bootstrap statistics.

        \FOR{$b = 1$ to $B$}
            \STATE Generate a random vector $V$ of size $n^{\dagger}$. Elements $v_i$ are i.i.d. and satisfy $\mathbb{E}(v)=0$ and $\mathrm{Var}(v) = 1$.  
            \STATE Compute the bootstrap residuals $\hat{\varepsilon}_{b}^{\dagger} = \hat{\varepsilon}^{\dagger} \cdot V$.
            \STATE Compute the adjusted residuals:
            \[
            \boldsymbol{\hat{\varepsilon}}_{p,b}^{\dagger} = \boldsymbol{\hat{\varepsilon}}_b^{\dagger} - \mathbf{\hat{g}}^{\dagger}((\mathbf{\hat{g}}^{\dagger})^\top \mathbf{\hat{g}}^{\dagger})^{-1}(\mathbf{\hat{g}}^{\dagger})^\top \boldsymbol{\hat{\varepsilon}}_b^{\dagger}.
            \]
            \STATE Compute the bootstrap statistic $\text{stat\_kerb}[b] = \sqrt{n^{\dagger}}(\hat{\mu}_{\hat{\theta}, \mathbb{S}^*,k_p})$ using $\boldsymbol{\hat{\varepsilon}}_{p,b}^{\dagger}$, $\mathbf{K}^{\dagger}$, trained dual coefficients $\boldsymbol{\alpha}^*$, shifted residuals that correspond to the support vector points: $\{\tilde{z}_{p,j}\}_{j \in \mathbb{S}^*}$ for SVM and $\{\tilde{\varepsilon}_{p,j}\}_{j \in \mathbb{S}^*}$ for OCSVM.
        \ENDFOR

        \STATE Compute the critical value $C_{\alpha,B}$ as the $(1-\alpha)$ quantile of $\text{stat\_kerb}$.
    \end{algorithmic}
\end{algorithm}

\section{More Simulation Results}
\subsection{OLS Simulation Results at Other Significance Levels}
We present simulation results at significance level $10\%$ and $1\%$ when the dimensions are $q=10$ (Tables \ref{tab:q10-10} and \ref{tab:q10-1}) and $q=20$ (Tables \ref{tab:q20-10} and \ref{tab:q20-1}).

\begin{table}[htbp]
    \centering
    \footnotesize
    \caption{Empirical sizes and powers at $10\%$ estimated by OLS with $q=10$}
    \begin{adjustbox}{width=\textwidth}
        \begin{tabular}{@{}lccccccccccc@{}}
            \toprule
            \multicolumn{1}{l}{$n$} & \multicolumn{5}{c}{200} & \multicolumn{6}{c}{400} \\
            \cmidrule(lr){2-6} \cmidrule(lr){8-12}
            & $\hat{T}_{\nu\text{-SVM}}$ & $\hat{T}_{OCSVM}$ & $\hat{T}_{GP}$ & $\hat{T}_{KCM}$ & $\hat{T}_{ICM}$ 
            & & $\hat{T}_{\nu\text{-SVM}}$ & $\hat{T}_{OCSVM}$ & $\hat{T}_{GP}$ & $\hat{T}_{KCM}$ & $\hat{T}_{ICM}$ \\
            
            \\
            \multicolumn{12}{c}{SIZE} \\
            \midrule
            $DGP_{1}$ (Bootstrap) & 0.133 & 0.115 & 0.079 & 0.002 & 0.000 && 0.091 & 0.099 & 0.087 & 0.008 & 0.001  \\
            \qquad \quad \,(Analytic) & [0.105] & [0.103] & - & - & - && [0.104] & [0.103] & - & - & -  \\
            
            \\
            \multicolumn{12}{c}{POWER} \\
            \midrule
            $DGP_{2}$ (Bootstrap) & 0.514 & 0.553 & 0.488 & 0.158 & 0.020 && 0.786 & 0.832 & 0.795 & 0.514 & 0.138  \\
            \qquad \quad \,(Analytic) & [0.533] & [0.503] & - & - & - && [0.827] & [0.791] & - & - & -  \\
            \\ 
            $DGP_{3}$ (Bootstrap) & 0.573 & 0.557 & 0.771 & 0.301 & 0.160 && 0.815 & 0.836 & 0.980 & 0.822 & 0.700  \\
            \qquad \quad \,(Analytic) & [0.560] & [0.574] & - & - & - && [0.826] & [0.801] & - & - & -  \\
            \\ 
            $DGP_{4}$ (Bootstrap) & 0.283 & 0.281 & 0.282 & 0.052 & 0.005 && 0.482 & 0.464 & 0.457 & 0.166 & 0.016  \\
            \qquad \quad \,(Analytic) & [0.278] & [0.281] & - & - & - && [0.488] & [0.454] & - & - & -  \\
            \\ 
            $DGP_{5}$ (Bootstrap) & 0.419 & 0.418 & 0.373 & 0.091 & 0.008 && 0.679 & 0.637 & 0.619 & 0.331 & 0.047  \\
            \qquad \quad \,(Analytic) & [0.398] & [0.411] & - & - & - && [0.688] & [0.668] & - & - & -  \\
            \bottomrule
        \end{tabular}
        \label{tab:q10-10}
    \end{adjustbox}
\end{table}

\begin{table}[htbp]
    \centering
    \footnotesize
    \caption{Empirical sizes and powers at $1\%$ estimated by OLS with $q=10$}
    \begin{adjustbox}{width=\textwidth}
        \begin{tabular}{@{}lccccccccccc@{}}
            \toprule
            \multicolumn{1}{l}{$n$} & \multicolumn{5}{c}{200} & \multicolumn{6}{c}{400} \\
            \cmidrule(lr){2-6} \cmidrule(lr){8-12}
            & $\hat{T}_{\nu\text{-SVM}}$ & $\hat{T}_{OCSVM}$ & $\hat{T}_{GP}$ & $\hat{T}_{KCM}$ & $\hat{T}_{ICM}$ 
            & & $\hat{T}_{\nu\text{-SVM}}$ & $\hat{T}_{OCSVM}$ & $\hat{T}_{GP}$ & $\hat{T}_{KCM}$ & $\hat{T}_{ICM}$ \\
            
            \\
            \multicolumn{12}{c}{SIZE} \\
            \midrule
            $DGP_{1}$ (Bootstrap) & 0.013 & 0.016 & 0.003 & 0.000 & 0.000 && 0.005 & 0.014 & 0.009 & 0.000 & 0.000  \\
            \qquad \quad \,(Analytic) & [0.013] & [0.010] & - & - & - && [0.017] & [0.006] & - & - & -  \\
            
            \\
            \multicolumn{12}{c}{POWER} \\
            \midrule
            $DGP_{2}$ (Bootstrap) & 0.188 & 0.202 & 0.165 & 0.013 & 0.000 && 0.466 & 0.463 & 0.491 & 0.163 & 0.004  \\
            \qquad \quad \,(Analytic) & [0.197] & [0.160] & - & - & - && [0.492] & [0.439] & - & - & -  \\
            \\ 
            $DGP_{3}$ (Bootstrap) & 0.197 & 0.228 & 0.335 & 0.038 & 0.002 && 0.498 & 0.522 & 0.792 & 0.453 & 0.082  \\
            \qquad \quad \,(Analytic) & [0.200] & [0.219] & - & - & - && [0.489] & [0.472] & - & - & -  \\
            \\ 
            $DGP_{4}$ (Bootstrap) & 0.062 & 0.061 & 0.042 & 0.001 & 0.000 && 0.177 & 0.149 & 0.156 & 0.032 & 0.000  \\
            \qquad \quad \,(Analytic) & [0.057] & [0.061] & - & - & - && [0.168] & [0.139] & - & - & -  \\
            \\ 
            $DGP_{5}$ (Bootstrap) & 0.126 & 0.120 & 0.096 & 0.007 & 0.000 && 0.324 & 0.286 & 0.289 & 0.079 & 0.005  \\
            \qquad \quad \,(Analytic) & [0.126] & [0.126] & - & - & - && [0.330] & [0.317] & - & - & -  \\
            \bottomrule
        \end{tabular}
        \label{tab:q10-1}
    \end{adjustbox}
\end{table}

\begin{table}[htbp]
    \centering
    \footnotesize
    \caption{Empirical sizes and powers at $10\%$ estimated by OLS with $q=20$}
    \begin{adjustbox}{width=\textwidth}
        \begin{tabular}{@{}lccccccccccc@{}}
            \toprule
            \multicolumn{1}{l}{$n$} & \multicolumn{5}{c}{200} & \multicolumn{6}{c}{400} \\
            \cmidrule(lr){2-6} \cmidrule(lr){8-12}
            & $\hat{T}_{\nu\text{-SVM}}$ & $\hat{T}_{OCSVM}$ & $\hat{T}_{GP}$ & $\hat{T}_{KCM}$ & $\hat{T}_{ICM}$ 
            & & $\hat{T}_{\nu\text{-SVM}}$ & $\hat{T}_{OCSVM}$ & $\hat{T}_{GP}$ & $\hat{T}_{KCM}$ & $\hat{T}_{ICM}$ \\
            
            \\
            \multicolumn{12}{c}{SIZE} \\
            \midrule
            $DGP_{1}$ (Bootstrap) & 0.116 & 0.124 & 0.024 & 0.000 & 0.000 && 0.104 & 0.097 & 0.017 & 0.000 & 0.000 \\
            \qquad \quad \,(Analytic) & [0.104] & [0.085] & - & - & - && [0.081] & [0.103] & - & - & - \\
            
            \\
            \multicolumn{12}{c}{POWER} \\
            \midrule
            $DGP_{2}$ (Bootstrap) & 0.350 & 0.368 & 0.138 & 0.000 & 0.000 && 0.571 & 0.583 & 0.269 & 0.003 & 0.000 \\
            \qquad \quad \,(Analytic) & [0.339] & [0.356] & - & - & - && [0.549] & [0.586] & - & - & - \\
            
            \\ 
            $DGP_{3}$ (Bootstrap) & 0.125 & 0.139 & 0.038 & 0.000 & 0.000 && 0.148 & 0.154 & 0.036 & 0.001 & 0.000 \\
            \qquad \quad \,(Analytic) & [0.128] & [0.116] & - & - & - && [0.153] & [0.149] & - & - & - \\
            
            \\ 
            $DGP_{4}$ (Bootstrap) & 0.638 & 0.633 & 0.390 & 0.000 & 0.000 && 0.900 & 0.889 & 0.679 & 0.060 & 0.000 \\
            \qquad \quad \,(Analytic) & [0.648] & [0.631] & - & - & - && [0.889] & [0.893] & - & - & - \\
            
            \\ 
            $DGP_{5}$ (Bootstrap) & 0.417 & 0.414 & 0.177 & 0.000 & 0.000 && 0.639 & 0.628 & 0.329 & 0.005 & 0.000 \\
            \qquad \quad \,(Analytic) & [0.357] & [0.354] & - & - & - && [0.636] & [0.625] & - & - & - \\
            \bottomrule
        \end{tabular}
        \label{tab:q20-10}
    \end{adjustbox}
\end{table}

\begin{table}[htbp]
    \centering
    \footnotesize
    \caption{Empirical sizes and powers at $1\%$ estimated by OLS with $q=20$}
    \begin{adjustbox}{width=\textwidth}
        \begin{tabular}{@{}lccccccccccc@{}}
            \toprule
            \multicolumn{1}{l}{$n$} & \multicolumn{5}{c}{200} & \multicolumn{6}{c}{400} \\
            \cmidrule(lr){2-6} \cmidrule(lr){8-12}
            & $\hat{T}_{\nu\text{-SVM}}$ & $\hat{T}_{OCSVM}$ & $\hat{T}_{GP}$ & $\hat{T}_{KCM}$ & $\hat{T}_{ICM}$ 
            & & $\hat{T}_{\nu\text{-SVM}}$ & $\hat{T}_{OCSVM}$ & $\hat{T}_{GP}$ & $\hat{T}_{KCM}$ & $\hat{T}_{ICM}$ \\
            
            \\
            \multicolumn{12}{c}{SIZE} \\
            \midrule
            $DGP_{1}$ (Bootstrap) & 0.010 & 0.011 & 0.000 & 0.000 & 0.000 && 0.013 & 0.013 & 0.000 & 0.000 & 0.000 \\
            \qquad \quad \,(Analytic) & [0.008] & [0.012] & - & - & - && [0.004] & [0.006] & - & - & - \\
            
            \\
            \multicolumn{12}{c}{POWER} \\
            \midrule
            $DGP_{2}$ (Bootstrap) & 0.084 & 0.112 & 0.000 & 0.000 & 0.000 && 0.237 & 0.229 & 0.014 & 0.000 & 0.000 \\
            \qquad \quad \,(Analytic) & [0.073] & [0.089] & - & - & - && [0.215] & [0.215] & - & - & - \\
            
            \\ 
            $DGP_{3}$ (Bootstrap) & 0.017 & 0.021 & 0.000 & 0.000 & 0.000 && 0.021 & 0.027 & 0.000 & 0.000 & 0.000 \\
            \qquad \quad \,(Analytic) & [0.013] & [0.014] & - & - & - && [0.014] & [0.029] & - & - & - \\
            
            \\ 
            $DGP_{4}$ (Bootstrap) & 0.287 & 0.264 & 0.007 & 0.000 & 0.000 && 0.657 & 0.638 & 0.140 & 0.000 & 0.000 \\
            \qquad \quad \,(Analytic) & [0.274] & [0.281] & - & - & - && [0.661] & [0.619] & - & - & - \\
            
            \\ 
            $DGP_{5}$ (Bootstrap) & 0.116 & 0.125 & 0.000 & 0.000 & 0.000 && 0.265 & 0.277 & 0.016 & 0.000 & 0.000 \\
            \qquad \quad \,(Analytic) & [0.089] & [0.080] & - & - & - && [0.285] & [0.273] & - & - & - \\
            \bottomrule
        \end{tabular}
        \label{tab:q20-1}
    \end{adjustbox}
\end{table}

\subsection{Finite Sample Performance of LASSO with $q=20$} 
We examine the finite sample performance of the proposed SVM-based test statistics at all three conventional significance levels when the estimator is the LASSO at dimension $q=20$.


\setlength{\colwidth}{(\textwidth - 5\tabcolsep - 5\arrayrulewidth)/9}

\begin{table}[htbp]
    \centering
    \caption{Empirical sizes and powers of $\hat{T}_{\nu\text{-SVM}}$ estimated by LASSO with $q=20$}
    {\footnotesize 
    \begin{tabular}{
        @{} 
        >{\raggedright\arraybackslash}p{0.18\textwidth} 
        *{6}{>{\centering\arraybackslash}p{\colwidth}} 
        @{}
    }
        \toprule
        \multicolumn{1}{@{}>{\raggedright\arraybackslash}p{0.18\textwidth}}{$n$} 
        & \multicolumn{3}{c}{$n = 200$} 
        & \multicolumn{3}{c}{$n = 400$} \\
        \cmidrule(lr){2-4} \cmidrule(lr){5-7}
        & $10\%$ & $5\%$ & $1\%$ & $10\%$ & $5\%$ & $1\%$ \\
        \midrule
        \multicolumn{7}{c}{$\hat{T}_{\nu\text{-SVM}}$ - SIZE} \\
        \midrule
        \multicolumn{1}{@{}>{\raggedright\arraybackslash}p{0.18\textwidth}}{$DGP_{1}$ (Bootstrap)} & 0.107 & 0.050 & 0.004 & 0.095 & 0.058 & 0.007 \\
        \multicolumn{1}{@{}>{\raggedright\arraybackslash}p{0.18\textwidth}}{\qquad \quad \,(Analytic)}      & [0.093] & [0.041] & [0.007] & [0.103] & [0.056] & [0.009] \\
        \addlinespace 
        \multicolumn{7}{c}{$\hat{T}_{\nu\text{-SVM}}$ - POWER} \\
        \midrule
        \multicolumn{1}{@{}>{\raggedright\arraybackslash}p{0.18\textwidth}}{$DGP_{2}$ (Bootstrap)} & 0.325 & 0.199 & 0.072 & 0.557 & 0.426 & 0.209 \\
        \multicolumn{1}{@{}>{\raggedright\arraybackslash}p{0.18\textwidth}}{\qquad \quad \,(Analytic)}      & [0.341] & [0.228] & [0.091] & [0.555] & [0.400] & [0.201] \\ \\
        \multicolumn{1}{@{}>{\raggedright\arraybackslash}p{0.18\textwidth}}{$DGP_{3}$ (Bootstrap)} & 0.131 & 0.058 & 0.017 & 0.149 & 0.080 & 0.024 \\
        \multicolumn{1}{@{}>{\raggedright\arraybackslash}p{0.18\textwidth}}{\qquad \quad \,(Analytic)}      & [0.116] & [0.053] & [0.010] & [0.150] & [0.082] & [0.028] \\  \\
        \multicolumn{1}{@{}>{\raggedright\arraybackslash}p{0.18\textwidth}}{$DGP_{4}$ (Bootstrap)} & 0.566 & 0.453 & 0.230 & 0.894 & 0.816 & 0.601 \\
        \multicolumn{1}{@{}>{\raggedright\arraybackslash}p{0.18\textwidth}}{\qquad \quad \,(Analytic)}      & [0.637] & [0.506] & [0.281] & [0.894] & [0.816] & [0.595] \\ \\
        \multicolumn{1}{@{}>{\raggedright\arraybackslash}p{0.18\textwidth}}{$DGP_{5}$ (Bootstrap)} & 0.384 & 0.274 & 0.100 & 0.609 & 0.494 & 0.272 \\
        \multicolumn{1}{@{}>{\raggedright\arraybackslash}p{0.18\textwidth}}{\qquad \quad \,(Analytic)}      & [0.403] & [0.279] & [0.130] & [0.618] & [0.488] & [0.269] \\
        \bottomrule
    \end{tabular}
    } 
    \label{tab:lasso-q20-nusvm}
\end{table}

\begin{table}[htbp]
    \centering
    \caption{Empirical sizes and powers of $\hat{T}_{OCSVM}$ estimated by LASSO with $q=20$}
    {\footnotesize 
    \begin{tabular}{
        @{} 
        >{\raggedright\arraybackslash}p{0.15\textwidth} 
        *{6}{>{\centering\arraybackslash}p{\colwidth}} 
        @{}
    }
        \toprule
        \multicolumn{1}{@{}>{\raggedright\arraybackslash}p{0.15\textwidth}}{$n$} 
        & \multicolumn{3}{c}{$n = 200$} 
        & \multicolumn{3}{c}{$n = 400$} \\
        \cmidrule(lr){2-4} \cmidrule(lr){5-7}
        & $10\%$ & $5\%$ & $1\%$ & $10\%$ & $5\%$ & $1\%$ \\
        \midrule
        \multicolumn{7}{c}{$\hat{T}_{OCSVM}$ - SIZE} \\
        \midrule
        \multicolumn{1}{@{}>{\raggedright\arraybackslash}p{0.18\textwidth}}{$DGP_{1}$ (Bootstrap)} & 0.103 & 0.047 & 0.010 & 0.108 & 0.054 & 0.010 \\
        \multicolumn{1}{@{}>{\raggedright\arraybackslash}p{0.18\textwidth}}{\qquad \quad \,(Analytic)}      & [0.095] & [0.044] & [0.005] & [0.109] & [0.054] & [0.010] \\
        \addlinespace 
        \multicolumn{7}{c}{$\hat{T}_{OCSVM}$ - POWER} \\
        \midrule
        \multicolumn{1}{@{}>{\raggedright\arraybackslash}p{0.18\textwidth}}{$DGP_{2}$ (Bootstrap)} & 0.336 & 0.210 & 0.074 & 0.556 & 0.434 & 0.197 \\
        \multicolumn{1}{@{}>{\raggedright\arraybackslash}p{0.18\textwidth}}{\qquad \quad \,(Analytic)}      & [0.358] & [0.259] & [0.109] & [0.533] & [0.404] & [0.190] \\ \\
        \multicolumn{1}{@{}>{\raggedright\arraybackslash}p{0.18\textwidth}}{$DGP_{3}$ (Bootstrap)} & 0.132 & 0.074 & 0.016 & 0.161 & 0.088 & 0.025 \\
        \multicolumn{1}{@{}>{\raggedright\arraybackslash}p{0.18\textwidth}}{\qquad \quad \,(Analytic)}      & [0.116] & [0.061] & [0.017] & [0.161] & [0.091] & [0.018] \\  \\
        \multicolumn{1}{@{}>{\raggedright\arraybackslash}p{0.18\textwidth}}{$DGP_{4}$ (Bootstrap)} & 0.637 & 0.515 & 0.263 & 0.887 & 0.813 & 0.599 \\
        \multicolumn{1}{@{}>{\raggedright\arraybackslash}p{0.18\textwidth}}{\qquad \quad \,(Analytic)}      & [0.634] & [0.510] & [0.260] & [0.891] & [0.820] & [0.597] \\ \\
        \multicolumn{1}{@{}>{\raggedright\arraybackslash}p{0.18\textwidth}}{$DGP_{5}$ (Bootstrap)} & 0.412 & 0.292 & 0.108 & 0.615 & 0.486 & 0.263 \\
        \multicolumn{1}{@{}>{\raggedright\arraybackslash}p{0.18\textwidth}}{\qquad \quad \,(Analytic)}      & [0.385] & [0.277] & [0.122] & [0.616] & [0.492] & [0.277] \\
        \bottomrule
    \end{tabular}
    } 
    \label{tab:lasso-q20-tocsvm}
\end{table}

\subsection{Simulation Results with DGPs with Intercept Terms}
In theory, including intercept terms in the data-generating process (DGP) should not affect test performance. However, in some DGPs, such an inclusion may lead to numerical instability in the inverse of the matrix $(\mathbf{\hat{g}}^\top \mathbf{\hat{g}})^{-1}$ when constructing the projection kernel. This numerical instability can result in reduced performance of the test statistics.

The following DGPs illustrate such cases:
\[
	\begin{split}
		& DGP_{1}^*: Y = \beta_0 + X \beta + \varepsilon, \\
		& DGP_{2}^*: Y = \beta_0 + X \beta + \|X\|_2 + \varepsilon, \\
		& DGP_{3}^*: Y = \beta_0 + X \beta + \|X\|_2 / \sqrt{n} + \varepsilon,
	\end{split}
\]
where $\varepsilon \sim N(0, 1)$, and the dimension of $X$ is $q = 10$ and $q = 20$. For $DGP_{1}^*$, when $q = 10$, $X_1, \ldots, X_5 \sim U(0, 1)$, while $X_6, \ldots, X_{10} \sim \mathcal{N}(0, 1)$. When $q = 20$, $X_1, \ldots, X_{10} \sim U(0, 1)$, while $X_{11}, \ldots, X_{20} \sim \mathcal{N}(0, 1)$.

For $DGP_{2}^*$ and $DGP_{3}^*$, when $q = 10$, $X_1, \ldots, X_5 \sim U(0, 1)$, while $X_6, \ldots, X_{10} \sim \mathcal{N}(0, 1 + 0.1 \cdot (i - 5))$ for $i = 6, \ldots, 10$. When $q = 20$, $X_1, \ldots, X_{10} \sim U(0, 1)$, while $X_{11}, \ldots, X_{20} \sim \mathcal{N}(0, 1 + 0.1 \cdot (i - 10))$ for $i = 11, \ldots, 20$.

We use the \texttt{numpy.linalg.solve} function instead of \texttt{numpy.linalg.inv} to compute the inverse of the matrix $(\mathbf{\hat{g}}^\top \mathbf{\hat{g}})^{-1}$, as it enhances numerical stability. Additionally, we adopt the same median heuristic Gaussian kernel parameter used in the main paper to construct the test statistics.

The simulation results are presented in Tables \ref{tab:intercept-q10-size} to \ref{tab:intercept-power-q20}. Numerical instability significantly degrades the performance of the test statistics. Key observations include: First, the analytic critical values for both OCSVM-based test statistics become unreliable. However, the bootstrap critical values remain robust and trustworthy. Second, only the SVM-based test statistics demonstrate accurate size control and maintain good power against the data-generating processes (DGPs) considered. 

\begin{table}[htbp]
    \centering
    \footnotesize
    \caption{Empirical sizes of intercept term models at $10\%$, $5\%$, and $1\%$ estimated by OLS with $q=10$}
    \begin{adjustbox}{width=\textwidth}
        \begin{tabular}{@{}lccccccccccc@{}}
            \toprule
            \multicolumn{1}{l}{$n$} & \multicolumn{5}{c}{200} & \multicolumn{6}{c}{400} \\
            \cmidrule(lr){2-6} \cmidrule(lr){8-12}
            & $\hat{T}_{\nu\text{-SVM}}$ & $\hat{T}_{OCSVM}$ & $\hat{T}_{GP}$ & $\hat{T}_{KCM}$ & $\hat{T}_{ICM}$ 
            & & $\hat{T}_{\nu\text{-SVM}}$ & $\hat{T}_{OCSVM}$ & $\hat{T}_{GP}$ & $\hat{T}_{KCM}$ & $\hat{T}_{ICM}$ \\
            
            \\
            \multicolumn{12}{c}{SIZE ($10\%$)} \\
            \midrule
            $DGP_{1}^*$ (Bootstrap) & 0.084 & 0.111 & 0.106 & 0.000 & 0.084 && 0.105 & 0.098 & 0.095 & 0.000 & 0.081 \\
            \qquad \quad \,(Analytic) & [0.002] & [0.000] & - & - & - && [0.002] & [0.000] & - & - & - \\
            
            \\
            \multicolumn{12}{c}{SIZE ($5\%$)} \\
            \midrule
            $DGP_{1}^*$ (Bootstrap) & 0.032 & 0.054 & 0.043 & 0.000 & 0.023 && 0.055 & 0.058 & 0.050 & 0.000 & 0.033 \\
            \qquad \quad \,(Analytic) & [0.000] & [0.000] & - & - & - && [0.000] & [0.000] & - & - & - \\
            
            \\
            \multicolumn{12}{c}{SIZE ($1\%$)} \\
            \midrule
            $DGP_{1}^*$ (Bootstrap) & 0.006 & 0.010 & 0.002 & 0.000 & 0.001 && 0.008 & 0.008 & 0.004 & 0.000 & 0.005 \\
            \qquad \quad \,(Analytic) & [0.000] & [0.000] & - & - & - && [0.000] & [0.000] & - & - & - \\
            \bottomrule
        \end{tabular}
        \label{tab:intercept-q10-size}
    \end{adjustbox}
\end{table}

\begin{table}[htbp]
    \centering
    \footnotesize
    \caption{Empirical powers of intercept term models at $10\%$, $5\%$, and $1\%$ estimated by OLS with $q=10$}
    \begin{adjustbox}{width=\textwidth}
        \begin{tabular}{@{}lccccccccccc@{}}
            \toprule
            \multicolumn{1}{l}{$n$} & \multicolumn{5}{c}{200} & \multicolumn{6}{c}{400} \\
            \cmidrule(lr){2-6} \cmidrule(lr){8-12}
            & $\hat{T}_{\nu\text{-SVM}}$ & $\hat{T}_{OCSVM}$ & $\hat{T}_{GP}$ & $\hat{T}_{KCM}$ & $\hat{T}_{ICM}$ 
            & & $\hat{T}_{\nu\text{-SVM}}$ & $\hat{T}_{OCSVM}$ & $\hat{T}_{GP}$ & $\hat{T}_{KCM}$ & $\hat{T}_{ICM}$ \\
            
            \\
            \multicolumn{12}{c}{POWER ($10\%$)} \\
            \midrule
            $DGP_{2}^*$ (Bootstrap) & 0.210 & 0.231 & 0.083 & 0.000 & 0.000 && 0.326 & 0.327 & 0.148 & 0.000 & 0.000 \\
            \qquad \quad \,(Analytic) & [0.010] & [0.010] & - & - & - && [0.019] & [0.022] & - & - & - \\ \\
            
            $DGP_{3}^*$ (Bootstrap) & 0.277 & 0.279 & 0.082 & 0.000 & 0.000 && 0.292 & 0.316 & 0.108 & 0.000 & 0.000 \\
            \qquad \quad \,(Analytic) & [0.031] & [0.031] & - & - & - && [0.020] & [0.037] & - & - & - \\
            
            \\
            \multicolumn{12}{c}{POWER ($5\%$)} \\
            \midrule
            $DGP_{2}^*$ (Bootstrap) & 0.134 & 0.139 & 0.026 & 0.000 & 0.000 && 0.217 & 0.221 & 0.063 & 0.000 & 0.000 \\
            \qquad \quad \,(Analytic) & [0.001] & [0.004] & - & - & - && [0.005] & [0.006] & - & - & - \\ \\
            
            $DGP_{3}^*$ (Bootstrap) & 0.181 & 0.180 & 0.018 & 0.000 & 0.000 && 0.204 & 0.188 & 0.039 & 0.000 & 0.000 \\
            \qquad \quad \,(Analytic) & [0.008] & [0.007] & - & - & - && [0.008] & [0.009] & - & - & - \\
            
            \\
            \multicolumn{12}{c}{POWER ($1\%$)} \\
            \midrule
            $DGP_{2}^*$ (Bootstrap) & 0.037 & 0.042 & 0.001 & 0.000 & 0.000 && 0.077 & 0.086 & 0.010 & 0.000 & 0.000 \\
            \qquad \quad \,(Analytic) & [0.000] & [0.000] & - & - & - && [0.000] & [0.001] & - & - & - \\ \\
            
            $DGP_{3}^*$ (Bootstrap) & 0.062 & 0.073 & 0.000 & 0.000 & 0.000 && 0.070 & 0.075 & 0.001 & 0.000 & 0.000 \\
            \qquad \quad \,(Analytic) & [0.000] & [0.001] & - & - & - && [0.001] & [0.000] & - & - & - \\
            \bottomrule
        \end{tabular}
        \label{tab:intercept-q10-power}
    \end{adjustbox}
\end{table}

\begin{table}[htbp]
    \centering
    \footnotesize
    \caption{Empirical sizes of intercept term models at $10\%$, $5\%$, and $1\%$ estimated by OLS with $q=20$}
    \begin{adjustbox}{width=\textwidth}
        \begin{tabular}{@{}lccccccccccc@{}}
            \toprule
            \multicolumn{1}{l}{$n$} & \multicolumn{5}{c}{200} & \multicolumn{6}{c}{400} \\
            \cmidrule(lr){2-6} \cmidrule(lr){8-12}
            & $\hat{T}_{\nu\text{-SVM}}$ & $\hat{T}_{OCSVM}$ & $\hat{T}_{GP}$ & $\hat{T}_{KCM}$ & $\hat{T}_{ICM}$ 
            & & $\hat{T}_{\nu\text{-SVM}}$ & $\hat{T}_{OCSVM}$ & $\hat{T}_{GP}$ & $\hat{T}_{KCM}$ & $\hat{T}_{ICM}$ \\
            
            \\
            \multicolumn{12}{c}{SIZE ($10\%$)} \\
            \midrule
            $DGP_{1}^*$ (Bootstrap) & 0.117 & 0.142 & 0.111 & 0.000 & 0.000 && 0.107 & 0.122 & 0.064 & 0.000 & 0.000 \\
            \qquad \quad \,(Analytic) & [0.001] & [0.001] & - & - & - && [0.000] & [0.001] & - & - & - \\
            
            \\
            \multicolumn{12}{c}{SIZE ($5\%$)} \\
            \midrule
            $DGP_{1}^*$ (Bootstrap) & 0.060 & 0.084 & 0.020 & 0.000 & 0.000 && 0.055 & 0.058 & 0.015 & 0.000 & 0.000 \\
            \qquad \quad \,(Analytic) & [0.000] & [0.000] & - & - & - && [0.000] & [0.000] & - & - & - \\
            
            \\
            \multicolumn{12}{c}{SIZE ($1\%$)} \\
            \midrule
            $DGP_{1}^*$ (Bootstrap) & 0.019 & 0.017 & 0.000 & 0.000 & 0.000 && 0.010 & 0.008 & 0.001 & 0.000 & 0.000 \\
            \qquad \quad \,(Analytic) & [0.000] & [0.000] & - & - & - && [0.000] & [0.000] & - & - & - \\
            \bottomrule
        \end{tabular}
        \label{tab:intercept-size-q20}
    \end{adjustbox}
\end{table}

\begin{table}[htbp]
    \centering
    \footnotesize
    \caption{Empirical powers of intercept term models at $10\%$, $5\%$, and $1\%$ estimated by OLS with $q=20$}
    \begin{adjustbox}{width=\textwidth}
        \begin{tabular}{@{}lccccccccccc@{}}
            \toprule
            \multicolumn{1}{l}{$n$} & \multicolumn{5}{c}{200} & \multicolumn{6}{c}{400} \\
            \cmidrule(lr){2-6} \cmidrule(lr){8-12}
            & $\hat{T}_{\nu\text{-SVM}}$ & $\hat{T}_{OCSVM}$ & $\hat{T}_{GP}$ & $\hat{T}_{KCM}$ & $\hat{T}_{ICM}$ 
            & & $\hat{T}_{\nu\text{-SVM}}$ & $\hat{T}_{OCSVM}$ & $\hat{T}_{GP}$ & $\hat{T}_{KCM}$ & $\hat{T}_{ICM}$ \\
            
            \\
            \multicolumn{12}{c}{POWER ($10\%$)} \\
            \midrule
            $DGP_{2}^*$ (Bootstrap) & 0.306 & 0.324 & 0.000 & 0.000 & 0.000 && 0.602 & 0.540 & 0.000 & 0.000 & 0.000 \\
            \qquad \quad \,(Analytic) & [0.110] & [0.121] & - & - & - && [0.310] & [0.311] & - & - & - \\ \\
            
            $DGP_{3}^*$ (Bootstrap) & 0.190 & 0.206 & 0.000 & 0.000 & 0.000 && 0.225 & 0.225 & 0.000 & 0.000 & 0.000 \\
            \qquad \quad \,(Analytic) & [0.053] & [0.068] & - & - & - && [0.054] & [0.059] & - & - & - \\
            
            \\
            \multicolumn{12}{c}{POWER ($5\%$)} \\
            \midrule
            $DGP_{2}^*$ (Bootstrap) & 0.213 & 0.228 & 0.000 & 0.000 & 0.000 && 0.452 & 0.431 & 0.000 & 0.000 & 0.000 \\
            \qquad \quad \,(Analytic) & [0.034] & [0.043] & - & - & - && [0.171] & [0.175] & - & - & - \\ \\
            
            $DGP_{3}^*$ (Bootstrap) & 0.116 & 0.119 & 0.000 & 0.000 & 0.000 && 0.133 & 0.153 & 0.000 & 0.000 & 0.000 \\
            \qquad \quad \,(Analytic) & [0.018] & [0.027] & - & - & - && [0.026] & [0.020] & - & - & - \\
            
            \\
            \multicolumn{12}{c}{POWER ($1\%$)} \\
            \midrule
            $DGP_{2}^*$ (Bootstrap) & 0.074 & 0.074 & 0.000 & 0.000 & 0.000 && 0.239 & 0.225 & 0.000 & 0.000 & 0.000 \\
            \qquad \quad \,(Analytic) & [0.000] & [0.004] & - & - & - && [0.023] & [0.026] & - & - & - \\ \\
            
            $DGP_{3}^*$ (Bootstrap) & 0.032 & 0.031 & 0.000 & 0.000 & 0.000 && 0.040 & 0.050 & 0.000 & 0.000 & 0.000 \\
            \qquad \quad \,(Analytic) & [0.001] & [0.002] & - & - & - && [0.001] & [0.002] & - & - & - \\
            \bottomrule
        \end{tabular}
        \label{tab:intercept-power-q20}
    \end{adjustbox}
\end{table}

\subsection{The Choice of $\nu$ in SVM Algorithms}
In both of the SVM algorithms, \(\nu\) controls the fraction of training errors allowed. We found that the choice of \(\nu\) does not significantly change the finite performance of the proposed test statistics, as illustrated in Figures \ref{fig:figure4} to \ref{fig:figure7}. However, this could be data-dependent.\\

\begin{figure}[htbp]
	\centering
    \begin{minipage}[b]{0.45\textwidth}
		\includegraphics[width=\textwidth]{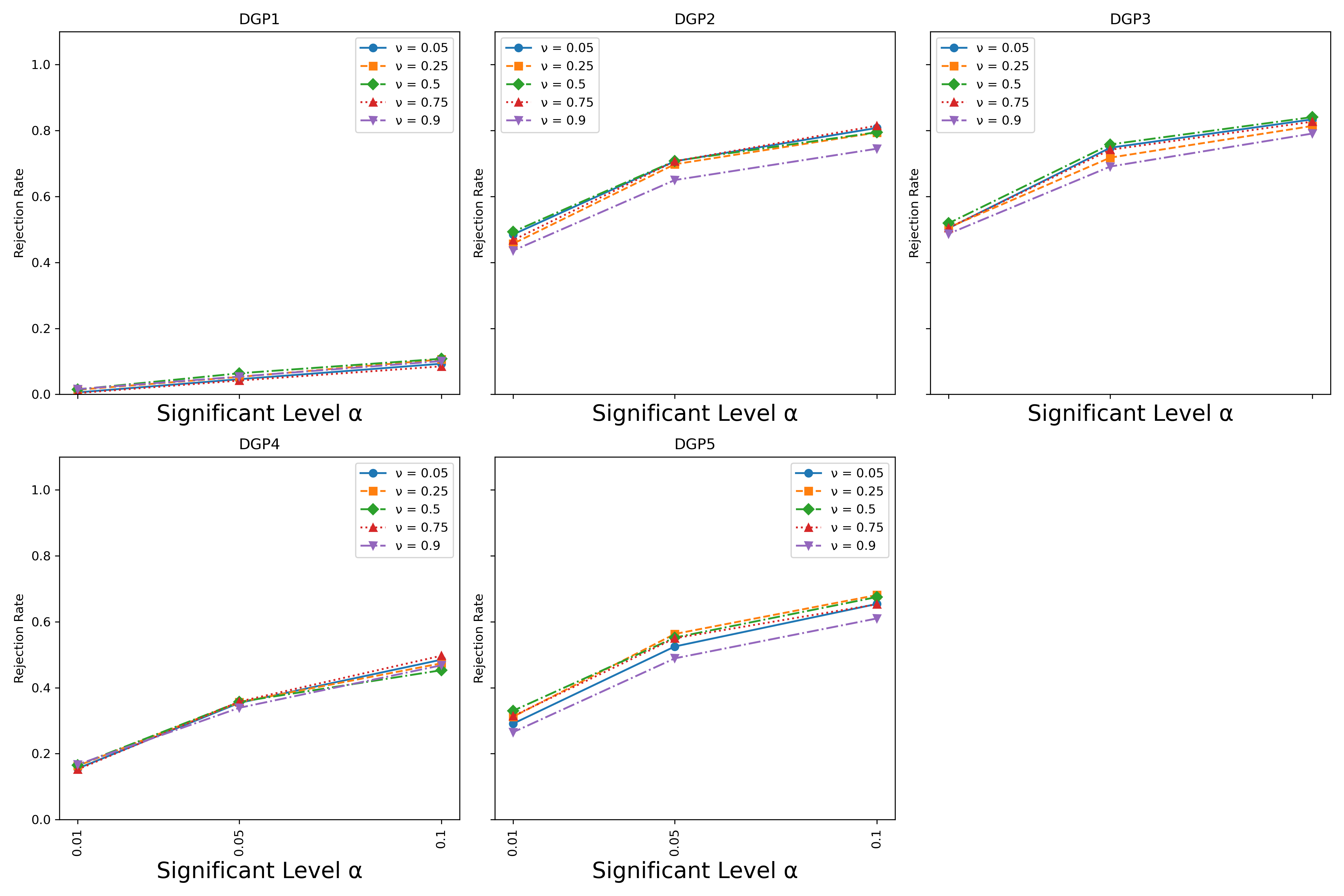}
		\caption{Sizes and Powers $\hat T_{\nu-SVM}$ with $q=10,N = 400$ at different $\nu$}
		\label{fig:figure4}
	\end{minipage}
	\hfill
	\begin{minipage}[b]{0.45\textwidth}
		\includegraphics[width=\textwidth]{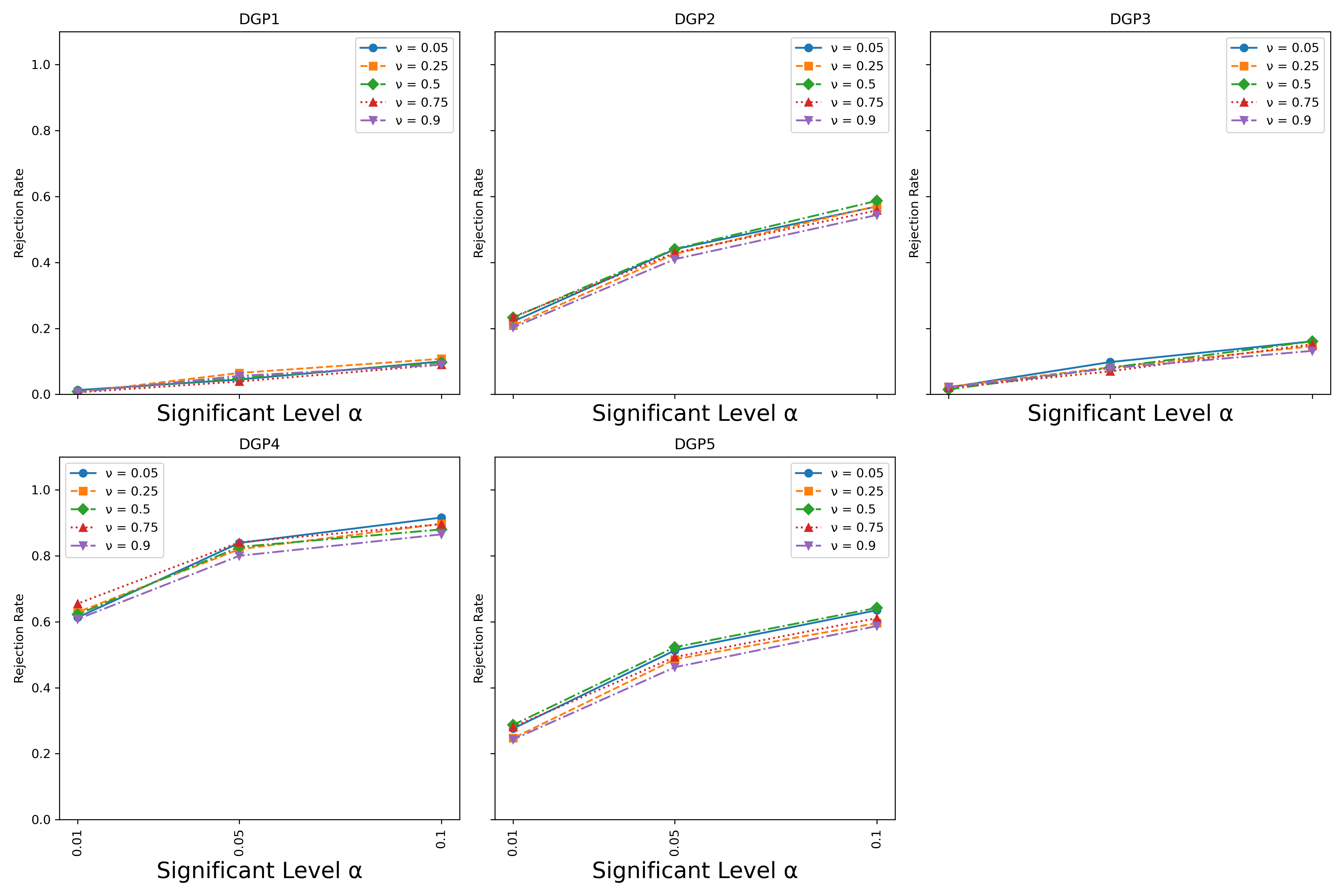}
		\caption{Sizes and Powers of $\hat T_{\nu-SVM}$ with $q=20,N = 400$ at different $\nu$}
		\label{fig:figure5}
	\end{minipage}

    \vspace{1em} 

	\begin{minipage}[b]{0.45\textwidth}
		\includegraphics[width=\textwidth]{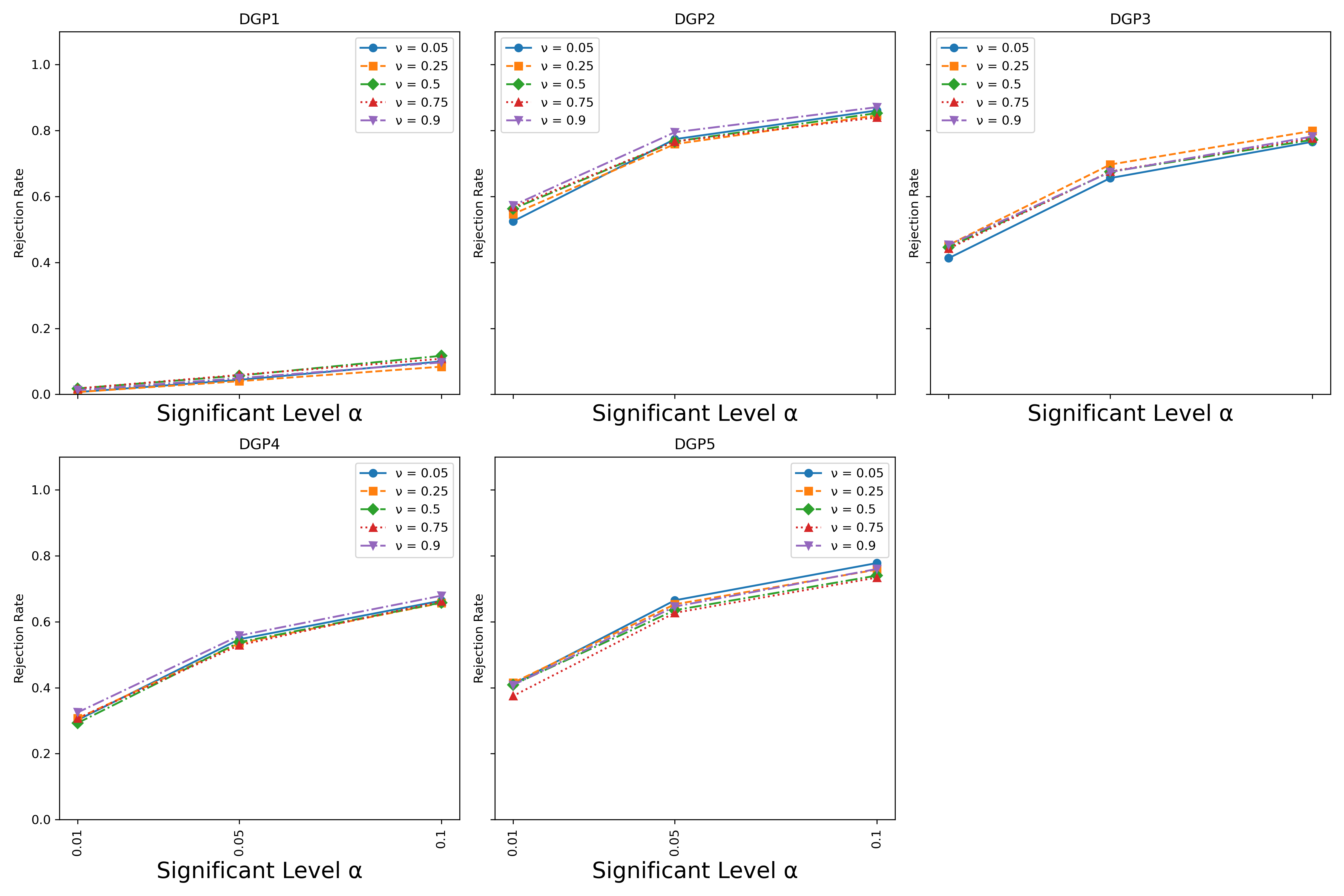}
		\caption{Sizes and Powers $\hat T_{OCSVM}$ with $q=10,N = 400$ at different $\nu$}
		\label{fig:figure6}
	\end{minipage}
	\hfill
	\begin{minipage}[b]{0.45\textwidth}
		\includegraphics[width=\textwidth]{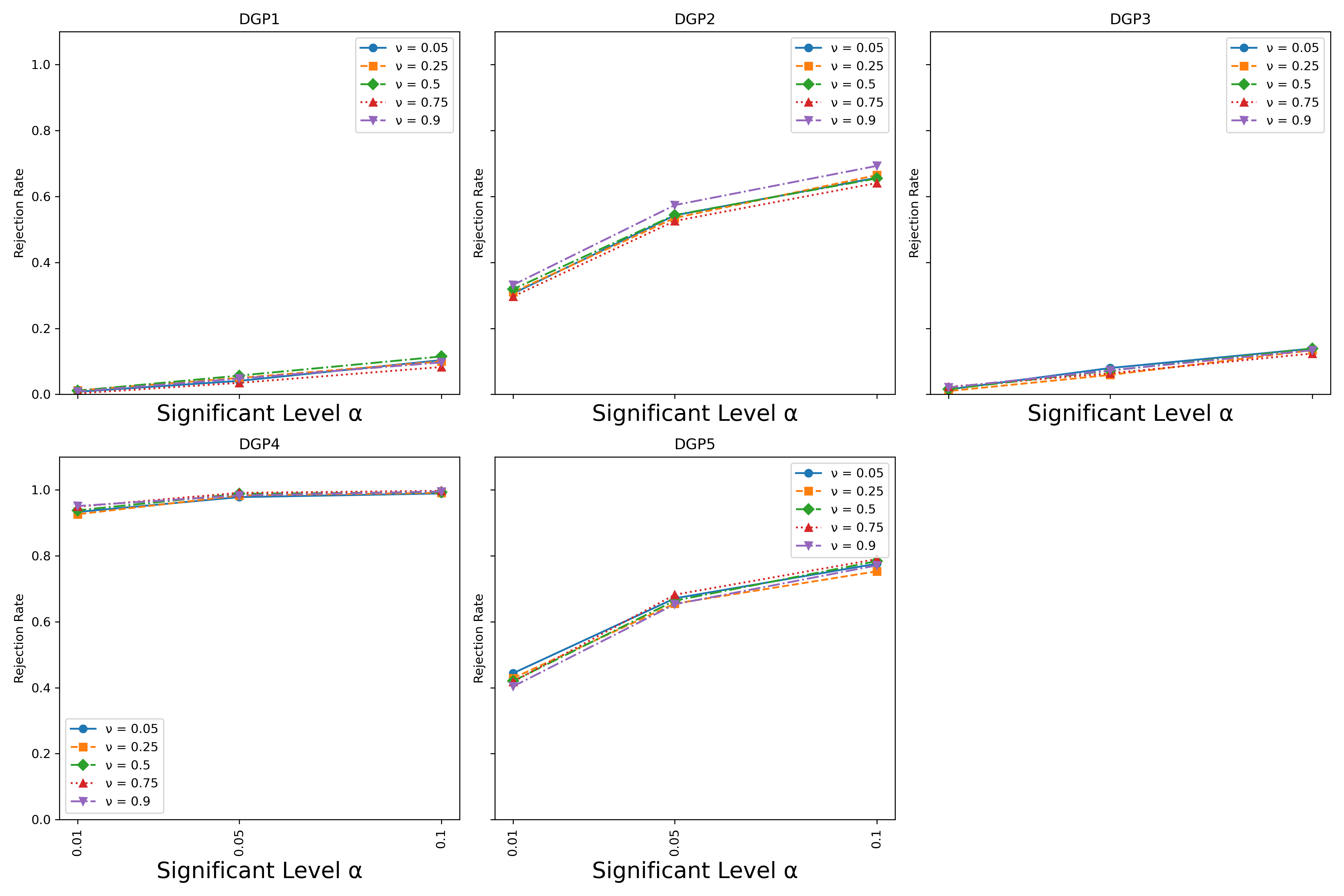}
		\caption{Sizes and Powers of $\hat T_{OCSVM}$ with $q=20,N = 400$ at different $\nu$}
		\label{fig:figure7}
	\end{minipage}
\end{figure}

\end{document}